\documentclass[12pt]{article}
\usepackage{graphicx}
\usepackage{amsmath,amssymb,amsthm,amsfonts}
\newtheorem{theorem}{Theorem}
\newtheorem{corollary}{Corollary}[section]
\newtheorem{lemma}{Lemma}[section]

\newtheorem{definition}{Definition}

\newtheorem{remark}{Remark}[section]

\numberwithin{equation}{section}
\begin{document}
\title
  {\bf A twisted integrable hierarchy with $\mathbb D_2$ symmetry
  }
\author   {  Derchyi Wu \\
{}\\
Institute of Mathematics, Academia Sinica\\
Taipei 10617, Taiwan\\
mawudc@math.sinica.edu.tw
}
\date{\today}
\maketitle
\section*{{Abstract}}
A loop algebra approach to the Gerdjikov-Mikhailov-Valchev (GMV) equation is provided to exploit the associated twisted integrable structure and a  new twisted integrable  hierarchy is discovered. Using the twisted loop algebra structure, we obtain a transparent treatment of the associated scattering and inverse scattering theory and solve the initial value problem for  the GMV equation. 

%%%%%%%%%%%%%%%%%%%%%%%%%%%%%%%%%%%%%%%%%%%%%%%%%%%%%%%%%%%%%%%%%%%%%%%%%%%%%%%%%%%%%%%%%%%%%%%%%%%%%%%%%%%%%%%%%%%%%%%%%%%%%%%%
\section{Introduction }\label{S:introduction}
%%%%%%%%%%%%%%%%%%%%%%%%%%%%%%%%%%%%%%%%%%%%%%%%%%%%%%%%%%%%%%%%%%%%%%%%%%%%%%%%%%%%%%%%%%%%%%%%%%%%%%%%%%%%%%%%%%%%%%%%%%%%%%%%%
Symmetry %$\mathfrak S$ 
is a novel  phenomenon in Nature.  It is also an important tools for scientists to unravel complicated dynamics. In the theory of integrable systems, studying  symmetries has been one of the central problems and yields rewarding results even beyond the field itself. Roughly speaking, two main approaches are adopted: (1) try to understand specific  problems by identifying %and investigating 
simple, symmetric structures that lie within them \cite{DS84}, \cite{T08}, etc; (2) try to classify integrable systems to provide a framework for ordering or understanding more general situation \cite{M81}, \cite{D90}, \cite{MSS91}, etc. 

One successful attempt in classification theory is the study  of reduction groups, formulated in \cite{M79}, \cite{M81}, and developed in \cite{M82}, \cite{MSY88}, \cite{GGK01}, \cite{GGK01-1}. Recent results have characterized Lax pairs with finite reduction groups of fractional-linear transformations, i.e., $\mathbb Z_N$, $\mathbb D_N$,  $\mathbb T$, $\mathbb O$ and $\mathbb I$ %(the additive group of integers modulo $N$, the group of a dihedron with $N$ corners, the tetrahedral groups, the octahedral groups, and the icosahedral groups) 
and aroused interest in the classsfication theory of automorphic Lie algebras \cite{LM04}, \cite{LM05}, \cite{LS10}.

Despite progress made in the classification theory of algebraic structures, the analytic properties such as the construction of solutions, the investigation of the inverse scattering theory, of the above  integrable systems remain mostly open. In particular, one of the simplest systems with $\mathbb D_2$-symmetry is  
the anisotropic deformation of a multicomponent generalization of the classical Heisenberg ferromagnetic equation:
\begin{equation}\label{E:GMV-itro}
\begin{split}
&i\vec{u}_t=(\vec u_x-\vec u(\vec u^*\cdot \vec u_x))_x+4\epsilon\vec u(\vec u^*\cdot J\vec u)+\textbf{A}\vec u,\\
&\vec u^* \vec u=1, \quad\vec u \in\mathbb C^{N-1},\quad 
J^2=1, \quad\left[\textbf{A},\,J\right]=0, \quad\epsilon>0.
\end{split}
\end{equation} Many interesting algebraic and analytic properties of (\ref{E:GMV-itro}) are provided in \cite{GMV10} but a complete resolution of the inverse problem and Cauchy problem is still demanded.

On the other hand, in studying  symmetries of the generalized sine-Gordon equations (GSGE), famous for being connected to submanifold geometry in Euclidean spaces, Terng introduced twisted $U/K$-hierarchies via a loop group approach \cite{T07}. The inverse scattering problem of one prototypical class of twisted $U/K$-hierarchies is then solved by encoding the loop algebra structures into the inverse scattering theory of GSGE \cite{ABT86} and  associated submanifold geometry in Minkowski spaces is derived \cite{MW11}. Twisted $U/K$-hierarchies are  integrable hierarchies with $\mathbb D_2$ symmetry.

Compared to the study of reduction groups, the loop group approach puts more emphasis on an organic assembling of ingredients of symmetries in integrable  systems \cite{A79}, \cite{W91}, \cite{TU00}, \cite{T07}, \cite{MW11}. To illustrate, given three involutions $\tau$, $\sigma_0$ and $\sigma_1$ on a simple Lie group, let $U$ be the real form of $\tau$, $U/K$ be the symmetric space defined by $\sigma_0$, $\mathcal U=\mathcal K+\mathcal P$, and $(\mathcal L_+,\mathcal L_-)$ be a splitting of the loop algebra $\mathcal L(\mathcal U/\mathcal K)$ % defined by $\sigma_1$ 
 such that
%\begin{equation}\label{E:involution}
\[\sigma_0(\xi(-\lambda))=\xi(\lambda),\quad\sigma_1(\xi(1/\lambda))=\xi(\lambda)\]
%\end{equation}
for $\xi\in \mathcal L_+$. 
Denote the corresponding projection map to $\mathcal L_\pm$ as $\hat\pi_\pm$. A twisted $U/K$-hierarchy is then defined by the Lax pair
\begin{equation}\label{E:twh}
\left[\partial_x+\hat\pi_{+}\left(m J_{1}m^{-1}\right),\, \partial_t+\hat\pi_{+}\left(m J_{j}m^{-1}\right)\right]=0,
\end{equation}
for some $m=m(x,t,\lambda)\in  L_{-}$, the loop group corresponding to the loop algebra $\mathcal L_-$, and constant loops $J_{ j}\in \mathcal L_+$ with coefficients in a Cartan subalgebra in $\mathcal P$. The loop group approach is rooted in, enhanced and enriched by the inverse scattering theory. 

The purpose of this paper is to provide a loop algebra approach to the anisotropic deformation of the multicomponent generalization of the Heisenberg ferromagnetic equation (\ref{E:GMV-itro}), $\textbf{N=3}$, called \textbf{\emph{the GMV equation}} for simplicity from now on, and to solve the inverse scattering theory. %For simplicity, from now on the anisotropic deformation of a multicomponent generalization of the classical Heisenberg ferromagnetic equation (\ref{E:GMV-itro}), $N=3$, is called the GMV equation. 
Distinct features discovered are:
\begin{itemize}
	\item The loop algebra factorization $(\mathcal L_+,\mathcal L_-)$ is not of a splitting type. Thus the twisted hierarchy associated with the GMV equation, i.e., the twisted $\frac {U(3)}{U(1)\times U(2)}$-hierarchy,  generalizes the twisted $U/K$-hierarchies defined in \cite{T07}.
	\item One needs to introduce an extended spectral problem and extended scattering data to solve the inverse problem. The extended spectral problem chosen is that for the twisted $\frac{U(4)}{U(2)\times U(2)}$-hierarchy which shares the same reduction group and can be "`projected"' to a twisted $\frac {U(3)}{U(1)\times U(2)}$-spectral problem when the scattering data is an extended one.
\end{itemize}

The paper is organized as follows: in Section \ref{S:TH}, we define the twisted $\frac {U(3)}{U(1)\times U(2)}$-hierarchy via a non-split factorization
 of the loop algebra and compute the explicit formula of a decisive coefficient, for the GMV equation, in the Lax pair.  Section \ref{S:GMV-lax} is
 the discussion of  the GMV equation and its relation with the twisted $\frac {U(3)}{U(1)\times U(2)}$-hierarchy. Section \ref{S:direct-GMV} and \ref{S:inverse-GMV} are
devoted to the scattering and inverse scattering theory of the twisted $\frac {U(3)}{U(1)\times U(2)}$-hierarchy. %In particular, we embed the original twisted $\frac {U(3)}{U(1)\times U(2)}$-spectral (direct) problem into that for the twisted $\frac{U(4)}{U(2)\times U(2)}$-hierarchy, derive scattering and inverse scattering theory for the twisted $\frac{U(4)}{U(2)\times U(2)}$-spectral problem in analogous to that for twisted $\frac{O(n,n)}{O(n)\times O(n)}$-spectral problem obtained in \cite{MW11}, and project the result of the inverse problem to that for the twisted $\frac {U(3)}{U(1)\times U(2)}$-hierarchy at last. 
The Cauchy problems of twisted $\frac {U(3)}{U(1)\times U(2)}$-flows and the GMV equation are solved in Section \ref{S:cauchy-GMV}.

We make two special remarks at last. Though the discussion in Section \ref{S:TH} and \ref{S:GMV-lax} should be extended for general $N$ by analogy, the spectral problem for (\ref{E:GMV-itro}) is no longer defined by an oblique direction \cite{ABT86}, \cite{MW11} and the associated direct problem cannot be solved when $N>3$. The other remark is a B$\ddot{\textrm{a}}$ckland transformation theory for the twisted $\frac{U(4)}{U(2)\times U(2)}$-flows should be obtained by adapting the amazing computation and theory for the GSGE \cite{BT88}. However, the extended scattering data is not preserved under these transformations. So the approach yields no  GMV solitons.

\noindent\textbf{\emph{Acknowledgements}} 

{The author would like to thank Professor Zixiang Zhou for initiating the work at Fudan University in November, 2011 and for many stimulating discusssions afterwards which made this work possible. The  author was  partially supported by NSC 100-2115-M-001-001. }

%%%%%%%%%%%%%%%%%%%%%%%%%%%%%%%%%%%%%%%%%%%%%%%%%%%%%%%%%%%%%%%%%%%%%%%%%%%%%%%%%%%%%%%%%%%%%%%%%%%%%%%%%%%%%%%%%%%%%%%%%%%%%%%%
\section{The twisted $\frac {U(3)}{U(1)\times U(2)}$-hierarchy }\label{S:TH}
%%%%%%%%%%%%%%%%%%%%%%%%%%%%%%%%%%%%%%%%%%%%%%%%%%%%%%%%%%%%%%%%%%%%%%%%%%%%%%%%%%%%%%%%%%%%%%%%%%%%%%%%%%%%%%%%%%%%%%%%%%%%%%%%%
Let $\sigma_i$, $i=1,\,2 $, be involutions on $U(3)$ defined by
\begin{equation}\label{E:sigma}
\begin{split}
\sigma_i(x)=J_{i}xJ_{i}^{-1},& \quad x\in U(3),\\
J_1=\textrm{diag}(1,-1,-1),&\quad J_2=\textrm{diag}(1,-1,1)\\
\end{split}
\end{equation}
and
$
{ u }(3)=\mathcal K_i \oplus\mathcal P_i$, $ i=1,\,2$, 
the Cartan decompositions for $\sigma_i$. Let $\mathcal K_i$ be the Lie algebras of $K_i$, i.e.,
\[
\begin{split}
&K_1=\{\left(\begin{array}{ccc}
a_{11} & 0 & 0 \\
0 & a_{22} & a_{23}\\
0 & a_{32} & a_{33}
\end{array}\right) :\ |a_{11}|=1,\ \left(\begin{array}{cc}
a_{22} & a_{23}\\
a_{32} & a_{33}
\end{array}\right)\in U(2)\, \},\\
&K_2=\{\left(\begin{array}{ccc}a_{11} & 0 & a_{13}\\ 0 & a_{22} & 0 \\ a_{31} & 0 & a_{33}\end{array}\right):\  |a_{22}|=1,\  \left(\begin{array}{cc}a_{11} & a_{13}\\ a_{31}  & a_{33}\end{array}\right)\in U(2)\,\},\\
&\mathcal P_1=\{\ i\left(\begin{array}{ccc}
0 & u & v \\
u^* & 0 & 0 \\
v^* & 0 & 0
\end{array}\right)\in u(3) \, \},\quad 
\mathcal P_2=\{\ i\left(\begin{array}{ccc}0 & u & 0 \\ u^*  & 0 & v \\ 0 & v^* & 0 \end{array}\right)\in u(3)\,\}.
\end{split}
\]Hence 
\begin{equation}\label{E:split-K}
\begin{split}
&S=K_1\cap K_2=\{\textrm{ diag}\,\left(
e^{i\alpha_1} , e^{i\alpha_2} , e^{i\alpha_3}\right)|\,\alpha_i\in\mathbb R\,\},\\
&\mathcal S=\{i\textrm{ diag}\,\left(
\alpha_1 ,\alpha_2,\alpha_3\right)|\,\alpha_i\in\mathbb R\,\},
\end{split}
\end{equation}
and
\begin{eqnarray}
&&K_1=S\times_{S_1} K_1',\ K_1'= 1\otimes SU(2),\label{E:split-K-1}\\
&&\mathcal K_1=\mathcal S +_{\mathcal S_1} \mathcal K_1',\ \mathcal K_1'= 0\oplus su(2),\label{E:split-K-2}\\
&&S_1=S\cap K_1'=\{\textrm{ diag}\,\left(
1 , e^{i\alpha} ,e^{-i\alpha}\right)|\,\alpha\in\mathbb R \},\label{E:split-S-1}\\
&&{\mathcal S}_1=\mathcal S\cap \mathcal K_1'=\{\,i\textrm{ diag}\,\left(
0 ,\alpha ,-\alpha\right)|\,\alpha\in\mathbb R \}.\label{E:split-S-2}
\end{eqnarray}
Here $K_1=S\times_{S_1} K_1'$ means for $\forall x\in K_1$, $x=\xi\eta$ with $\xi\in S$, $\eta\in K_1'$ and if
\[\begin{split}
&x=\xi\eta=\tilde \xi\tilde \eta, \quad \xi,\,\tilde \xi\in S,\quad\eta,\, \tilde \eta\in K_1',
\end{split}
\] then $\xi^{-1}\tilde\xi=\eta\tilde\eta^{-1}\in S_1$. By analogy $\mathcal K_1=\mathcal S+_{\mathcal S_1} \mathcal K_1'$ is defined by factoizations of elements in $\mathcal K_1$ up to factors in $\mathcal S_1$. 

Furthermore, for a fixed $\epsilon>0$, define the loop groups 
\begin{eqnarray*}
L^\epsilon&=&\{f:{\mathfrak A}_{\sqrt\epsilon\delta,\sqrt\epsilon/\delta}\stackrel{holo.}{\rightarrow}GL_{3}(\mathbb C)\,|\, \left(f(\bar\lambda)\right)^\ast f(\lambda)=I,\,\sigma_1(f(-\lambda))=f(\lambda)\}\\
%L^{\sigma_0}  &=&\left\{f\in L|\  \sigma_0(f(-\lambda))=f(\lambda)\right\},\\
L_{+}^\epsilon  & = & \{f\in { L} \,|\, \sigma_2\left(f(\epsilon/\lambda)\right)=f(\lambda) \, \}, \\
L_{-}^\epsilon    &= &\{f\in {L}\,|\,f:\mathbb C/\mathfrak D_{\sqrt\epsilon\delta}\stackrel{holo. }{\rightarrow}GL_{3}(\mathbb C),\,f(\infty)\in K_1'\,  \}.
\end{eqnarray*}
Here $0<\delta<1$, $\mathfrak S^{r}$ is the circle of radius $r$ centered at $0$, $\mathfrak D_{r}$ is the disk of radius $r$, and $\mathfrak A_{r_1,r_2}$ is the annulus with boundaries $\mathfrak S^{r_1}$ and $\mathfrak S^{r_2}$.  Then the Lie algebras of $L^\epsilon$, $ L_{+}^\epsilon$, $ L_{-}^\epsilon$ are
\begin{eqnarray*}
{\mathcal L}^\epsilon&=& \{\xi(\lambda)=\sum_{j\le n_0}\xi_j\lambda^j |\ \textit{$\xi_j\in\mathcal K_1$ if $j$ is even, $\xi_j\in\mathcal P_1$ if $j$ is odd}\},\\
{\mathcal L}_{+}^\epsilon  &=&   \{\xi(\lambda)=\sum_{|j|\le n_0}\xi_j\lambda^j\in \mathcal L^\epsilon |\ \xi_{-j}=\sigma_2(\xi_j)\epsilon^j,\, \xi_0\in\mathcal S \},\\
{\mathcal L}_{-}^\epsilon &=&  \{\xi(\lambda)=\sum_{j\le 0}\xi_j\lambda^j\in \mathcal L^\epsilon| \ \xi_0\in\mathcal K_1' \}.
\end{eqnarray*}
Similarly, we have a non-splitting decomposition $\mathcal L^\epsilon={\mathcal L}_{+}^\epsilon +_{\mathcal S_1}{\mathcal L}_{-}^\epsilon$ and can define projections $\hat\pi_{\pm}$ of $\xi\in \mathcal L^\epsilon$ onto $ {\mathcal L}_{+}^\epsilon $, ${\mathcal L}_-^\epsilon$, up to factors in $\mathcal S_1$,  by the following relations:
\begin{equation}\label{E:pipm}
\begin{split}
&\hat\pi_{+}(\xi) =\pi_{\mathcal S}(\xi_0)%\pi_{\mathcal S_0}(\xi_0)-\pi_{\mathcal S_1}\left(\sum_{0<j,\textit{ $j$ even}}(\,\xi_j+\sigma_1(\xi_j)\,)\right)+
+\sum_{0<j\le n_0}\left(\xi_j\lambda^j+\sigma_2(\xi_j)\left(\frac{\epsilon}{\lambda}\right)^j\right),\\
&\hat\pi_{-}(\xi) =\pi_{\mathcal K_1'}(\xi_0)%+\pi_{\mathcal S_1}\left(\sum_{0<j,\textit{ $j$ even}}(\,\xi_j+\sigma_1(\xi_j)\,)\right)
+\sum_{0<j\le n_0}(\,\xi_{-j} -\sigma_2(\xi_j)\epsilon^j\,)\lambda^{-j},\\
&\xi=\hat\pi_{+}(\xi)+_{\mathcal S_1}\hat\pi_{-}(\xi),\quad \xi_0=\pi_{\mathcal S}(\xi_0)+_{\mathcal S_1}\pi_{\mathcal K_1'}(\xi_0).
\end{split}
\end{equation}
Let $\xi(\lambda)\sim_{\mathcal S_1}\tilde\xi(\lambda)$ mean that $\xi(\lambda)-\tilde\xi(\lambda)$ is a constant loop in $\mathcal S_1$. 
Finally, let
\begin{equation}\label{E:CSA}
\mathcal A=\{i\left(\begin{array}{ccc} d_1 & r & 0 \\ r & d_1 & 0\\ 0 & 0 & d_3\end{array}\right)\in  u(3): r,\,d_1, \, d_3\in\mathbb R\}
\end{equation}
be a maximal abelian subalgebra in $u(3)$,
\begin{equation}\label{E:a}
a=i\left(\begin{array}{ccc}
0 & 1 & 0 \\
1 & 0 & 0 \\
0 & 0 & 0
\end{array}\right)\ \in \ \mathcal P_1\cap \mathcal A,
\end{equation}and
\begin{gather}
%m(x,t,\lambda)=b(x,t)m'(x,t,\lambda),\label{E:cauchy-1}\\
%\Psi(x,t,\lambda)=m(x,t,\lambda)e^{x\hat J_1+t\hat J_2},\label{E:normalization}\\
\hat J_{1,0}= a\lambda+\sigma_2(a)(\frac\epsilon\lambda)\ \in \mathcal P_1\cap \mathcal A\cap \mathcal L^\epsilon_+,\label{E:j1-0}\\
\hat J_k= i^{{k-1}}a^k\lambda^k
-i\left(\begin{array}{ccc}d_1 & 0 & 0 \\ 0 & d_1 & 0 \\ 0 & 0 & d_3 \end{array}\right)
+i^{{k-1}}\sigma_2(a^k)(\frac\epsilon\lambda)^{k}\ \in  \mathcal A\cap \mathcal L^\epsilon_+,\label{E:j-0}
\end{gather}for $k\in\{1,2,\dots\}$. Thus we have the commutativity condition 
\begin{equation}
[\hat J_{1,0},\hat J_k]=0.\label{E:commute}
\end{equation}
\begin{definition}\label{D:TH}
The $k$-th twisted ${\frac {U(3)}{U(1)\times U(2)}}$-flow, parametrized by $(d_1, d_3)$, in the twisted ${\frac {U(3)}{U(1)\times U(2)}}$-hierarchy is the compatibility condition 
\begin{eqnarray}
&&\left[{\bf L}, {\bf M}\right]=0,
\label{E:j-thflow}
\end{eqnarray}where
\begin{eqnarray}
&&{\bf L}=\partial_x-\frac{\partial\Psi}{\partial x}\Psi^{-1}=\partial_x-(\lambda bab^{-1}+\frac \epsilon\lambda\sigma_2(bab^{-1})),\label{E:j-thflow-L}\\
&&{\bf M}=\partial_t-\frac{\partial\Psi}{\partial t}\Psi^{-1},\label{E:j-thflow-M}\\
&&\Psi(x,t,\lambda)=m(x,t,\lambda)e^{x\hat J_{1,0}+t\hat J_k},\label{E:j-thflow-psi}
\end{eqnarray}
for some $m=m(x,t,\cdot)\in  L_{-}^\epsilon$ and $b(x,t)=m(x,t,\infty)\in \mathfrak P_{1}$ or $\mathfrak P_{2}$. Here
\begin{eqnarray}
&&\mathfrak P_{1}=\{f:R^2\to K_1'|f(\cdot,t)- \left(\begin{array}{ccc} 1 & 0 & 0 \\ 0 & 1 &0 \\ 0 & 0 & 1  \end{array}\right)\in\mathbb S,\textit{ $\forall t$}\},\label{E:potential-1}\\
&&\mathfrak P_{2}=\{f:R^2\to K_1'|f(\cdot,t)- \left(\begin{array}{ccc} 1 & 0 & 0 \\ 0 & 0 & -1 \\ 0 & 1 & 0  \end{array}\right)\in\mathbb S,\textit{ $\forall t$}\},\label{E:potential-2}
\end{eqnarray}
and $\mathbb S$ is the Schwartz space.
\end{definition}
We remark that the inverse scattering theory derived in this report shows $m(x,t,\lambda)$ is determined by  $\partial_x^jb(x,t)$. Hence the $k$-th twisted ${\frac {U(3)}{U(1)\times U(2)}}$-flow is a system in $b(x,t)$ and $b(x,t)$ is also called the $k$-th twisted ${\frac {U(3)}{U(1)\times U(2)}}$-flow if no ambiguity occurs. Theorem \ref{T:cauchy} will provide the existence theorem by solving the initial value problem of the $k$-th twisted ${\frac {U(3)}{U(1)\times U(2)}}$-flows.  
 
By the definition of twisted ${\frac {U(3)}{U(1)\times U(2)}}$-flows, one has
\begin{eqnarray*}
\frac{\partial \Psi}{\partial x}\Psi^{-1}
%&=&\left\{\frac{\partial}{\partial x}\left(m(x,t,\lambda)e^{x(\lambda a+\frac 1\lambda\sigma_2(a)\,)+t(i\lambda^2 a^2+\tilde{\textbf{A}}+i\frac 1{\lambda^2}\sigma_2(a^2)\,)}\right)\right\}\Psi^{-1}\nonumber\\
&=&\left\{\left[\frac{\partial m}{\partial x}+m\left(\lambda a+\frac \epsilon{\lambda}\sigma_2(a)\right)\right]e^{x\hat J_{1,0}+t\hat J_k}\right\} e^{-x\hat J_{1,0}-t\hat J_k}m^{-1}\\
&=&\frac{\partial m}{\partial x}m^{-1}+m\left(\lambda a+\frac \epsilon{\lambda}\sigma_2(a)\right)m^{-1}\\
&\sim_{\mathcal S_1}&\hat\pi_+\left(m\hat J_{1,0}m^{-1}\right).%\label{E:spectral-0-0}
\end{eqnarray*}
Similarly,  
$\frac{\partial \Psi}{\partial t}\Psi^{-1}\in\mathcal L_+^\epsilon$ as well and 
\begin{eqnarray}
\frac{\partial \Psi}{\partial t}\Psi^{-1}
&=&\sum_1^k\left(P_j\lambda^j+\sigma_2(P_j)(\frac \epsilon\lambda)^j\right)+P_0\label{E:first-p-j}\\
%&=&\left\{\left[\frac{\partial m}{\partial t}+m\hat J_k\right]e^{x\hat J_1+t\hat J_k}\right\}\cdot e^{-x\hat J_1-t\hat J_k}m^{-1}\nonumber\\
&=&\frac{\partial m}{\partial t}m^{-1}+m\hat J_km^{-1},\label{E:first-p-j-more}\\
&\sim_{\mathcal S_1}&\hat\pi_+\left(m\hat J_km^{-1}\right)\nonumber
%\label{E:spectral-0-0-1}
\end{eqnarray}
Thus the $k$-th twisted ${\frac {U(3)}{U(1)\times U(2)}}$-flows satisfy
\begin{equation}\label{E:cong-twisted-0}
\begin{split}
{\bf L}&\sim_{\mathcal S_1}\partial_x-\hat\pi_+ (m\hat J_{1,0}m^{-1} ) ,\\
{\bf M}&\sim_{\mathcal S_1}\partial_t-\hat\pi_+ (m\hat J_km^{-1} )
\end{split}
\end{equation}
which are similar to (\ref{E:twh}). 

\begin{lemma}\label{L:firstflow}
The first twisted ${\frac {U(3)}{U(1)\times U(2)}}$-flow is  the linear system
\[
\frac{\partial }{\partial x}\left(bab^{-1}\right)-\frac{\partial }{\partial t}\left(bab^{-1}\right)=\left[ bab^{-1},i\left(\begin{array}{ccc}c_1 & 0 & 0 \\ 0 & c_2 & 0 \\ 0 & 0 & c_3\end{array}\right)\right]
\]
where $c_i$ are real constants.
\end{lemma}
\begin{proof} Following (\ref{E:j-thflow}), (\ref{E:j-thflow-L}), and (\ref{E:first-p-j}), we obtain
\begin{gather}
\frac{\partial }{\partial x}\left(bab^{-1}\right)-\frac{\partial }{\partial t}\left(bab^{-1}\right)-\left[bab^{-1},P_0\right]=0,\label{E:1-coeffi-first}\\
\frac{\partial }{\partial x}P_0-\left[bab^{-1},\epsilon\sigma_2(bab^{-1})\right]
-\left[\epsilon\sigma_2(bab^{-1}),bab^{-1}\right]=0.\label{E:0-coeffi-first}
\end{gather}Thus $P_0$ is independent of $x$ (and $\lambda$).  Therefore $P_0$ is constant by taking the limit of (\ref{E:first-p-j-more}) when $x\to -\infty$,  $\lambda\to\infty$,  and $m(x=-\infty,t,\lambda=\infty)= b(x=-\infty,t)=1$ or $\left(\begin{array}{ccc} 1 & 0 & 0 \\ 0 & 0 & -1 \\ 0 & 1 & 0  \end{array}\right)$.
\end{proof}

We proceed to characterizing the $k$-th twisted ${\frac {U(3)}{U(1)\times U(2)}}$-flow, $k\ge 2$, by showing that  the coefficients $P_j$, $j\ne 0
$, defined by (\ref{E:first-p-j}), of ${\bf M}$ can be computed explicitly in terms of $x$-derivatives of entries of $b(x,t)$. 
As for $P_0$, owing to the non-splitting factorization $\mathcal L^\epsilon=\mathcal L^\epsilon_++_{\mathcal S_1}\mathcal L_-^\epsilon$ (up to factors in $\mathcal S_1$), only the first diagonal entry of $P_0$ can be computed explicitly (in terms of $x$-derivatives of entries of $b(x,t)$). The last two diagonal entries of $P_0$ are (inexplicit) functions in $\partial_x^jb(x,t)$ as we have remarked earlier.  This phenomenon is distinct from that of twisted flows defined in \cite{MW11}, \cite{T07}.

\begin{lemma}\label{L:lax}
For the $k$-th twisted ${\frac {U(3)}{U(1)\times U(2)}}$-flow, 
\begin{gather}
m(\partial_x-\hat J_{1,0})m^{-1} =\partial_x-(bab^{-1}\lambda+\sigma_2(bab^{-1})\frac \epsilon\lambda),\label{E:lax-1}\\
m(\partial_t-\hat J_k)m^{-1} =\partial_t-\sum_1^k\left(P_j\lambda^j+\sigma_2(P_j)(\frac \epsilon\lambda)^j\right)-P_0.\label{E:lax-2}
\end{gather}
Here $P_i$ are defined by (\ref{E:first-p-j}).
\end{lemma}
\begin{proof}
By (\ref{E:j-thflow-L}), we have
\[
\left(\partial_x\Psi\right)\Psi^{-1}=bab^{-1}\lambda+\sigma_2(bab^{-1})\frac\epsilon\lambda.
\]
So (\ref{E:j-thflow-psi}) implies 
\[
(\partial_xm+m\hat J_{1,0})m^{-1}=bab^{-1}\lambda+\sigma_2(bab^{-1})\frac\epsilon\lambda,
\]
which is equivalent to (\ref{E:lax-1}). The identity (\ref{E:lax-2}) can be proved similarly.
\end{proof}

\begin{lemma}\label{L:coeff}
For the $k$-th twisted ${\frac {U(3)}{U(1)\times U(2)}}$-flow, 
\begin{equation}\label{E:recursive}
[\partial_x-bab^{-1}\lambda-\sigma_2(bab^{-1})\frac\epsilon\lambda, m\hat J_km^{-1}]=0.
\end{equation}
\end{lemma}
\begin{proof} Because  $\hat J_k$ are loops with constant coefficients, by (\ref{E:commute}), we have
\[
[\partial_x-\hat J_{1,0},\hat J_k]=0,
\]which implies
\[
[m(\partial_x-\hat J_{1,0})m^{-1}, m(\hat J_k)m^{-1}]=0.
\]
By (\ref{E:lax-1}), we derive (\ref{E:recursive}).
\end{proof}

\begin{lemma}\label{L:coeff-complete}
The coefficients $P_j$,  $j\ne 0
$, and the first diagonal element of $P_0$ of ${\bf M}$ are fixed functions of components of $\partial_x^{s} b$, $0\le s \le k-j$. 
\end{lemma}

By the decomposition property (\ref{E:pipm}),  Lemma \ref{L:coeff}, and defining
\begin{equation}\label{E:coeff-j2}
m\hat J_km^{-1}=\sum_1^k\left(P_j\lambda^j+\sigma_2(P_j)(\frac \epsilon\lambda)^j\right)+P_0+\sum_{j=0}^\infty R_j\lambda^{-j},
\end{equation}
 we obtain the recursive formula on $P_j$, $0<j\le k$:
\begin{equation}\label{E:recursive-use}
\begin{split}
&P_k=i^{k-1}ba^kb^{-1}\\
&\partial_xP_k-\left[bab^{-1},P_{k-1}\right]=0,\\
&\partial_xP_{k-1}-\left[bab^{-1},P_{k-2}\right]-\left[\epsilon\sigma_2(bab^{-1}), P_k\right]=0,\\
&\vdots\\
&\partial_xP_{2}-\left[bab^{-1},P_{1}\right]-\left[\epsilon\sigma_2(bab^{-1}), P_3\right]=0,\\
&\partial_xP_{1}-\left[bab^{-1},P_{0}+R_0\right]-\left[\epsilon\sigma_2(bab^{-1}), P_2\right]=0.
\end{split}
\end{equation}
 Therefore, we can adapt the argument of the proof of Lemma 2.3 in \cite{MW11} to prove Lemma \ref{L:coeff-complete}. We skip the proof. Instead, we compute the case  $k=2$  for the purpose of solving the Cauchy problem of the GMV equation in this report.
\begin{lemma}\label{L:p10}
Write
\begin{equation}\label{E:bmu}
 b(x,t)=m(x,t,\infty)=\left(\begin{array}{ccc}
1 & 0 & 0\\
 0 & u & -\bar v\\
 0 & v & \bar u
\end{array}
\right)\in \mathfrak P_1 \cup \mathfrak P_2,\quad\vec u=\left(\begin{array}{c}u\\v\end{array}\right),
\end{equation}
then 
\begin{equation}\label{E:coefficient-p1}
P_1
=\left(\begin{array}{cc}
0 & -\left((1-\vec u\vec u^*) \vec u_x \right)^*\\
(1-\vec u\vec u^*) \vec u_x & 0_{2\times 2}
\end{array}
\right)
\end{equation}
for the second twisted ${\frac {U(3)}{U(1)\times U(2)}}$-flow.
\end{lemma}
\begin{proof}We first define $T=
\left(\begin{array}{rrr}
-\frac 1{\sqrt 2}& 0 & \frac 1{\sqrt 2}\\
\frac 1{\sqrt 2} & 0 & \frac 1{\sqrt 2}\\
 0 & 1 & 0
\end{array}
\right)$, so
\begin{gather}
T^{-1}aT=
\left(\begin{array}{rrr}
-i & 0 & 0\\
0 & 0 & 0\\
 0 & 0 & i
\end{array}
\right),\label{E:convenience-2}\\
bT
=
\left(\begin{array}{rrr}
-\frac 1{\sqrt 2}&  0 & \frac 1{\sqrt 2}\\
\frac 1{\sqrt 2} u & -\bar v & \frac 1{\sqrt 2}u \\
\frac 1{\sqrt 2} v & \bar u & \frac 1{\sqrt 2}v
\end{array}
\right),\ \left(bT\right)^{-1}
=
\left(\begin{array}{rrr}
-\frac 1{\sqrt 2}&   \frac 1{\sqrt 2}\bar u &   \frac 1{\sqrt 2}\bar v\\
0& -v & u \\
\frac 1{\sqrt 2}  & \frac 1{\sqrt 2}\bar u & \frac 1{\sqrt 2} \bar v
\end{array}
\right),\label{E:terminology}\\ 
P_2=-i
\left(\begin{array}{rrr}
1 & 0 & 0\\
0 & |u|^2 & u \bar v\\
 0 & \bar u v & |v|^2
\end{array}
\right), \quad \partial_x P_2=-i
\left(\begin{array}{rrr}
0 & 0 & 0\\
0 & |u|^2_x & (u \bar v)_x\\
 0 & (\bar uv)_x & |v|^2_x
\end{array}
\right).\label{E:convenience-1}
\end{gather}

By (\ref{E:recursive-use}), we have
\begin{eqnarray}
&&[bab^{-1},P_1]=\partial_xP_2.\label{E:2-coeff}%\\
%&&[bab^{-1},P_0+R_0]=\partial_xP_1-[\sigma_2(bab^{-1}),P_2].\label{E:1-coeff}
\end{eqnarray}
Taking the conjuation $(bT)^{-1}\,\cdot\,(bT)$ on both sides of (\ref{E:2-coeff}) and using (\ref{E:bmu}), (\ref{E:convenience-2})-(\ref{E:convenience-1}), we obtain
\[
[\left(\begin{array}{ccc}
-i & 0 & 0\\
0 & 0 & 0\\
 0 & 0 & i
\end{array}
\right),(bT)^{-1}P_1(bT)]
=-\frac i{\sqrt 2}
\left(\begin{array}{ccc}
0 & \bar u\bar v_x-\bar u_x\bar v  & 0\\
uv_x-u_xv& 0 & uv_x-u_xv\\
 0 & \bar u\bar v_x-\bar u_x\bar v & 0
\end{array}
\right).
\]
Thus the off diagonal part of $(bT)^{-1}P_1(bT)$, denoted as $\left[(bT)^{-1}P_1(bT)\right]^{o}$, is
\begin{equation}\label{E:off}
\left[(bT)^{-1}P_1(bT)\right]^{o}
=
\frac 1{\sqrt 2}
\left(\begin{array}{ccc}
0 & \bar u\bar v_x-\bar u_x\bar v & 0\\
-(uv_x-u_xv) & 0 & uv_x-u_xv\\
 0 & -(\bar u\bar v_x-\bar u_x\bar v) & 0
\end{array}
\right).
\end{equation}

On the other hand, using the minimal polynomial of $(bT)^{-1}m\hat J_2m^{-1}(bT)$ is 
\[
(X+i(\lambda^2+d_1+\frac {\epsilon^2}{\lambda^2})I) (X+id_3 I),
\] we obtain
\begin{equation}\label{E:minimal}
\begin{split}
&[
\left(\begin{array}{ccc}
-i & 0 & 0\\
0 & 0 & 0\\
 0 & 0 & -i
\end{array}
\right)\lambda^2+(bT)^{-1}P_1(bT)\lambda+(bT)^{-1}(P_0+R_0)m^{-1}(bT)+\cdots \\
&+i(\lambda^2+d_1+\frac{\epsilon^2}{\lambda^2})I ]\times 
[
\left(\begin{array}{ccc}
-i & 0 & 0\\
0 & 0 & 0\\
 0 & 0 & -i
\end{array}
\right)\lambda^2+(bT)^{-1}P_1(bT)\lambda\\
&+(bT)^{-1}(P_0+R_0)m^{-1}(bT)+\cdots +id_3I ]=0.
\end{split}
\end{equation}Equating the $\lambda^3$-coefficient of (\ref{E:minimal}) yields the diagonal part of $(bT)^{-1}P_1(bT)$, denoted as $\left[(bT)^{-1}P_1(bT)\right]^{d}$, which is $0$. Therefore, $(bT)^{-1}P_1(bT)=\left[(bT)^{-1}P_1(bT)\right]^{o}$. Together with (\ref{E:terminology}) and (\ref{E:off}), we obtain
\[
\begin{split}
P_1
=&
\left(\begin{array}{ccc}
0 &  v(\bar u\bar v_x-\bar u_x\bar v )& -u(\bar u\bar v_x-\bar u_x\bar v )\\
-\bar v(uv_x-u_xv) & 0 & 0\\
\bar u(uv_x-u_xv) & 0 & 0
\end{array}
\right)\\
=&\left(\begin{array}{cc}
0 & -\left((1-\vec u\vec u^*) \vec u_x \right)^*\\
(1-\vec u\vec u^*) \vec u_x & 0_{2\times 2}.
\end{array}
\right).
\end{split}
\]
\end{proof}
%%%%%%%%%%%%%%%%%%%%%%%%%%%%%%%%%%%%%%%%%%%%%%%%%%%%%%%%%%%%%%%%%%%%%%%%%%%%%%%%%%%%%%%%%%%%%%%%%%%%%%%%%%%%%%%%%%%%%%%%%%%%%%%%
\section{The GMV equation }\label{S:GMV-lax}
%%%%%%%%%%%%%%%%%%%%%%%%%%%%%%%%%%%%%%%%%%%%%%%%%%%%%%%%%%%%%%%%%%%%%%%%%%%%%%%%%%%%%%%%%%%%%%%%%%%%%%%%%%%%%%%%%%%%%%%%%%%%%%%%%
In studying integrable systems with reductions, one of the simplest nontrivial systems introduced by Gerdjikov, Mikhailov, Valchev \cite{GMV10}, \cite{GMV10-1},  \cite{GGMV12}, is  
the anisotropic multicomponent generalization of the classical Heisenberg ferromagnetic equation:
\begin{equation}\label{E:GMV}
\begin{split}
i\vec{u}_t=&(\vec u_x-\vec u(\vec u^*\cdot \vec u_x))_x+4\epsilon\vec u(\vec u^*\cdot J\vec u)+\textbf{A}\vec u,
\end{split}
\end{equation}
where 
\begin{equation}\label{E:NLE}
\begin{split}
\vec u^* \vec u=1,& \quad \vec u(x,t)\in\mathbb C^2,\\%=(u_1, u_2)^T, \\
J=\left(\begin{array}{cc} -1 & 0 \\ 0 & 1 \end{array}\right),&\quad\textbf{A}=\left(\begin{array}{cc}
 \alpha & 0 \\  0 & \beta \end{array}\right), \ \alpha,\,\beta\in\mathbb R.
\end{split}
\end{equation}
%$u=u(x,t)\in\mathbb C$, $v=v(x,t)\in\mathbb C$, and $\alpha$, $\beta$ are real constants. 
The equation (\ref{E:GMV}),  called the GMV equation for simplicity, has a $\mathbb Z_2\times \mathbb Z_2$ reduced Lax representation 
\begin{eqnarray}
&&\hskip1.2in \left[{\bf L},{\bf M}\right] =0,\label{E:lax}\\
{\bf L}&=&\partial_ x- bab^{-1} \lambda
-\sigma_2(bab^{-1}) \frac\epsilon\lambda,\label{E:GMV-lax-bab-1}\\
{\bf M}&=&\partial_t-iba^2b^{-1}\lambda^2-p_1\lambda-p_0-\sigma_2( p_1)\frac \epsilon\lambda-i\sigma_2(ba^2b^{-1})(\frac\epsilon\lambda)^{2},
\label{E:GMV-lax-bab-2}
\end{eqnarray}
with
\begin{eqnarray}
&&a=i\left(\begin{array}{ccc}
0 & 1 & 0\\
 1 & 0 & 0\\
 0 & 0 & 0
\end{array}
\right) \in \mathcal A,\quad
b(x,t)=\left(\begin{array}{ccc}
1 & 0 & 0\\
 0 & u & -\bar v\\
 0 & v & \bar u
\end{array}
\right)\in  K_1',\label{E:GMV-lax-bab}\\
&&p_1(x,t)= \left(\begin{array}{cc} 0 & -\vec a^* \\ \vec a & 0 \end{array}\right)\in\mathcal P_1,\quad \vec a=(1-\vec u\vec u^*)\vec u_x,\quad \vec u=\left(\begin{array}{c}u\\v\end{array}\right),\label{E:GMV-lax-bab-0-new}\\
&&p_0 =- i\left(\begin{array}{cc} -2\epsilon\vec u^* J\vec u  & 0 \\ 0 & \epsilon\left(J\vec u\vec u^*+\vec u\vec u^* J\right) \end{array}\right)-i\textrm{ diag}\,(0,\alpha,\beta)\ \in\mathcal S.
\label{E:GMV-lax-bab-0}
\end{eqnarray} 
It is readily to see that  ${\bf L}-\partial_x\in\mathcal L_+^\epsilon$, ${\bf M}-\partial_t\in\mathcal L_+^\epsilon$. 

\begin{lemma}\label{L:loop-algebra}
Suppose %we have a compatibility condition
\begin{equation}\label{E:lax-unique}
\left[{\bf L}, {\bf M}'\right] =0,
\end{equation}
where ${\bf L}$ is defined by (\ref{E:GMV-lax-bab-1}), (\ref{E:GMV-lax-bab}), and ${\bf M}'-\partial_t\in\mathcal L_+^\epsilon$ with $-iba^2b^{-1}\lambda^2$ as its leading term.
Then there exist real functions $\gamma(x,t)$, $\alpha_1(t)$, $\alpha_2(t)$ and $\alpha_3(t)$, such that
\begin{equation}
{\bf M}'=\partial_t-iba^2b^{-1}\lambda^2-p'_1\lambda-p'_0-\sigma_2(p'_1)\frac\epsilon\lambda-i\sigma_2(ba^2b^{-1})(\frac\epsilon\lambda)^2,\label{E:lax-p-lemma}
\end{equation}
with 
\begin{eqnarray}
&&p'_1-p_1= \gamma b ab^{-1}\in\mathcal P_1, \label{E:lax-q-2-ast}\\%\quad \vec\theta(x,t)=(a_1,a_2)^T,
&&p'_0+i\left(\begin{array}{cc}-2\epsilon\vec u^*J\vec u & 0 \\
0 & \epsilon\left(J\vec u\vec u^*+\vec u\vec u^*J\right)\end{array}\right) =
-i\,\textrm{diag }(\alpha_1,\alpha_2,\alpha_3),\label{E:lax-q-3-lemma-ast}
\end{eqnarray}
and  $p_1$ being the coefficient of $\bf M$ defined by (\ref{E:GMV-lax-bab-0-new}). Moreover, 
\begin{equation}\label{E:GMV-modified}
 \begin{split}
&i\vec{u}_t=(\vec u_x-\vec u(\vec u^*\cdot \vec u_x)+i\gamma\vec u)_x+4\epsilon\vec u(\vec u^*\cdot J\vec u)+\textbf{A}'\vec u ,\\
&\textbf{A}'=\textrm{diag }(\alpha,\ \beta),\ \alpha=\alpha_2(t)-\alpha_1(t),\ \beta=\alpha_3(t)-\alpha_1(t). 
\end{split}
\end{equation}
\end{lemma}
\begin{proof}%It suffices to identify coefficients of ${\bf M}$ and ${\bf M}'$. 
By assumption, one can set
\begin{eqnarray}
{\bf L} &=& \partial_x-q_1\lambda-\sigma_2(q_1)\frac\epsilon\lambda,\label{E:lax-q}\\
{\bf M} '&=& \partial_t-p'_2\lambda^2-p'_1\lambda-p'_0-\sigma_2(p'_1)\frac\epsilon\lambda-\sigma_2(p'_2)(\frac\epsilon\lambda)^2,\label{E:lax-p}
\end{eqnarray}
with
\begin{eqnarray}
q_1&=&bab^{-1},\quad
p_2'=iba^2b^{-1},\label{E:lax-q-1}\\
p'_1&=& \left(\begin{array}{cc} 0 & -\vec \theta^* \\ \vec \theta & 0 \end{array}\right)\in\mathcal P_1, \quad \vec\theta(x,t)=\left(\begin{array}{c}\theta_1\\ \theta_2\end{array}\right),\label{E:lax-q-2}\\
p'_0 &=&- i\left(\begin{array}{cc} -2\epsilon\vec u^* J\vec u  & 0 \\ 0 & \epsilon\left(J\vec u\vec u^*+\vec u\vec u^* J \right)\end{array}\right)
-i\textrm{ diag }(\alpha_1,\alpha_2,\alpha_3),\label{E:lax-q-3}
\end{eqnarray}
and $\alpha_i=\alpha_i(x,t)\in\mathbb R$. The compatibility condition (\ref{E:lax-unique}) then yields
\begin{eqnarray}
\partial_xp'_2-[q_1,p'_1]&=&0,\label{E:coefficient-2}\\
\partial_xp'_1-\partial_tq_1-[q_1, p'_0]-[\epsilon\sigma_2(q_1),p'_2] &=&0,\label{E:coefficient-1}\\
\partial_xp'_0-[q_1,\epsilon\sigma_2(p'_1)]-[\epsilon\sigma_2(q_1), p'_1]&=&0.\label{E:coefficient-0}
\end{eqnarray}

Let $\vec u=\left(\begin{array}{c}u \\ v\end{array}\right)=\left(\begin{array}{c} u_1 \\ u_2\end{array}\right)$. Then (\ref{E:coefficient-2}) implies
%\[
%\begin{split}
%&\bar u_1a_1+\bar u_2a_2+u_1\bar a_1+u_2\bar a_2=0,\\
%&\partial_x(\vec u\vec u^*)
%-\left(\begin{array}{cc} u_1\bar a_1 & u_1\bar a_2 \\ u_2\bar a_1 & u_2\bar a_2 \end{array}\right)
%-\left(\begin{array}{cc} \bar u_1  a_1 & \bar u_2 a_1 \\ \bar u_1 a_2& \bar u_2 a_2\end{array}\right)=0
%\end{split}
%\]or
%\begin{eqnarray}
%&&\bar u_1a_1+\bar u_2a_2+u_1\bar a_1+u_2\bar a_2=0,\label{E:eq1}\\
%&&\partial_x(|u_1|^2)-u_1\bar a_1-\bar u_1  a_1=0,\label{E:eq2}\\
%&&\partial_x(|u_2|^2)-u_2\bar a_2-\bar u_2 a_2=0,\label{E:eq3}\\
%&&\partial_x(u_1\bar u_2)-u_1\bar a_2-\bar u_2 a_1=0,\label{E:eq4}\\
%&&\partial_x(u_2\bar u_1)-u_2\bar a_1-\bar u_1 a_2=0.\label{E:eq5}
%\end{eqnarray}
%Hence 
\begin{eqnarray}
\bar u_1\theta_1+\bar u_2\theta_2+u_1\bar \theta_1+u_2\bar \theta_2&=&0,\label{E:eq1-1}\\
\partial_x(|u_1|^2)-u_1\bar \theta_1-\bar u_1  \theta_1&=&0,\label{E:eq2-1}\\
\partial_x(u_1\bar u_2)-u_1\bar \theta_2-\bar u_2 \theta_1&=&0.\label{E:eq3-1}
\end{eqnarray}
It is an under-determined linear system.  One then obtains 
\begin{equation}\label{E:p-a}
\vec \theta=(1-\vec u\vec u^*)\vec u_x+i\gamma(x,t)\vec u,\quad \gamma(x,t)\in\mathbb R.
\end{equation}
On the other hand, (\ref{E:coefficient-1}), %reduces to
%\begin{eqnarray*}
%0&=&-\partial_x\left(\begin{array}{cc} 0 & -\vec \theta^* \\ \vec \theta & 0 \end{array}\right)
%+i\partial_t\left(\begin{array}{cc} 0 & \vec u^* \\ \vec u & 0 \end{array}\right)\\
%&&+\left[-i\left(\begin{array}{cc} 0 & \vec u^* \\ \vec u & 0 \end{array}\right), i\textrm{ diag }(0,\alpha',-\alpha')
%+i\left(\begin{array}{cc} -2\vec u^*J\vec u & 0 \\ 0 & J\vec u\vec u^*+ \vec u\vec u^*J\end{array}\right)\right]\\
%&&+\left[-i\left(\begin{array}{cc} 0 & \vec u^* J\\ J\vec u & 0 \end{array}\right),
%i\left(\begin{array}{cc} 1 & 0 \\ 0 & \vec u\vec u^*\end{array}\right)
%\right]
%\end{eqnarray*}which by 
(\ref{E:eq1-1})-(\ref{E:eq3-1}) imply
\begin{equation}\label{E:NLE-r}
-\partial_x\vec \theta+i\partial_t\vec u-\textbf{A}'\vec u-2\epsilon\vec u\vec u^*J\vec u-2\epsilon\vec u\vec u^*J\vec u=0
\end{equation}
with $\textbf{A}'=\textrm{diag }(\alpha(x,t),\beta(x,t))$,  $\alpha=\alpha_2-\alpha_1$, and $\beta=\alpha_3-\alpha_1$. Hence we obtain 
\begin{equation}\label{E:GMV-modify}
\begin{split}
i\vec{u}_t=&(\vec u_x-\vec u(\vec u^*\cdot \vec u_x)+i\gamma\vec u)_x+4\epsilon\vec u(\vec u^*\cdot J\vec u)+\textbf{A}'\vec u,
\end{split}
\end{equation}
 by (\ref{E:p-a}). Moreover, (\ref{E:coefficient-0})) is equivalent to
%\begin{eqnarray*}
%0&=&i\partial_x\left(\textrm{ diag }(0,\alpha',-\alpha')
%+\left(\begin{array}{cc} -2\vec u^*J\vec u & 0 \\ 0 & J\vec u\vec u^*+ \vec u\vec u^*J\end{array}\right)\right)\\
%&&+\left[-i\left(\begin{array}{cc} 0 & \vec u^* \\ \vec u & 0 \end{array}\right), -\left(\begin{array}{cc} 0 & -\vec \theta^* J\\ J\vec \theta & 0 \end{array}\right)\right]\\
%&&+\left[-i\left(\begin{array}{cc} 0 & \vec u^* J\\ J\vec u & 0 \end{array}\right),
%-\left(\begin{array}{cc} 0 & -\vec \theta^* \\ \vec \theta & 0 \end{array}\right)\right]
%\end{eqnarray*}
%That is,
\begin{eqnarray}
\epsilon^{-1}\partial_x\alpha_1-2\partial_x(-|u_1|^2+|u_2|^2)+2\vec u^*J\vec \theta+2\vec \theta^*J\vec u&=&0,\label{E:11}\\
\partial_x  \textbf{A}''+J\left(\partial_x(\vec u\vec u^*)-(\vec \theta\vec u^*+\vec u\vec \theta^*)\right)+\left(\partial_x(\vec u\vec u^*)-(\vec \theta\vec u^*+\vec u\vec \theta^*)\right)J&=&0,\label{E:22}
\end{eqnarray}
with $\textbf{A}''=\epsilon^{-1}\textrm{diag }(\alpha_2(x,t),\alpha_3(x,t))$. Note (\ref{E:eq1-1}) and (\ref{E:eq2-1}) imply
\[
-2\partial_x(-|u_1|^2+|u_2|^2)+2\vec u^*J\vec \theta+2\vec \theta^*J\vec u=0.
\]Together with (\ref{E:11}) yields $\alpha_1=\alpha(t)$. Besides,
\begin{eqnarray}
&&\partial_x(\vec u\vec u^*)-(\vec \theta\vec u^*+\vec u\vec \theta^*)\label{E:aux}\\
=&&\vec u_x\vec u^*-\vec \theta\vec u^*+\vec u\vec u^*_x-\vec u\vec \theta^*\nonumber\\
=&&\vec u_x\vec u^*-\left[\left(1-\vec u\vec u^*\right)\vec u_x+i\gamma\vec u\right]\vec u^*+\vec u\vec u^*_x-\vec u\left[\vec u^*_x\left(1-\vec u\vec u^*\right)-i\gamma\vec u^*\right]\nonumber\\
=&&\vec u_x\vec u^*-\left(1-\vec u\vec u^*\right)\vec u_x\vec u^*+\vec u\vec u^*_x-\vec u\vec u^*_x\left(1-\vec u\vec u^*\right)\nonumber\\
=&&\vec u\vec u^*\vec u_x\vec u^*+\vec u\vec u_x^*\vec u\vec u^*\nonumber\\
=&&\vec u\left(|\vec u|^2\right)_x\vec u^*\nonumber\\
=&& 0.\nonumber
\end{eqnarray}
Here we have used (\ref{E:p-a}) and $b\in 1\times SU(2)$. Combining (\ref{E:22}) and (\ref{E:aux}), we obtained $\alpha_i=\alpha_i(t)$, $i=2,\,3$. % The reality of $\alpha_i$ comes from $P'_0\in 0\oplus su(2)$. Taking the limits of $P_0'$ in (\ref{E:m'}), we obtain $\alpha'=\alpha$. That is, $\textbf{A}=\textbf{A}'$ and $P_0=P_0'$. 
\end{proof}
Given $\epsilon>0$, for the GMV equation parametrized by $(4\epsilon,\beta)$ or $(\alpha,-4\epsilon)$ with  arbitary real constants $\alpha$, $\beta$, we have
\begin{theorem}\label{T:TH-GMV}
Write the second twisted ${\frac {U(3)}{U(1)\times U(2)}}$-flow $
 b(x,t)=\left(\begin{array}{ccc}
1 & 0 & 0\\
 0 & u & -\bar v\\
 0 & v & \bar u
\end{array}
\right)$, then $\vec u=\left(\begin{array}{c}u\\v\end{array}\right)$ satisfies the GMV equation
\begin{equation}\label{E:GMV-modified-TH-GMV}
 \begin{split}
&i\vec{u}_t=(\vec u_x-\vec u(\vec u^*\cdot \vec u_x))_x+4\epsilon\vec u(\vec u^*\cdot J\vec u)+\textbf{A}'\vec u ,\\
&\textbf{A}'=\left(\begin{array}{cc} 4\epsilon & 0 \\ 0 & 2\epsilon+d_3-d_1\end{array}\right)\textit{ or }\left(\begin{array}{cc} -2\epsilon+d_3-d_1 & 0 \\ 0 & -4\epsilon\end{array}\right). 
\end{split}
\end{equation}
\end{theorem}
\begin{proof} By Definition \ref{D:TH} and Lemma \ref{L:loop-algebra}, the second twisted ${\frac {U(3)}{U(1)\times U(2)}}$-flow satisfies 
\[
 \begin{split}
&i\vec{u}_t=(\vec u_x-\vec u(\vec u^*\cdot \vec u_x)+i\gamma\vec u)_x+4\epsilon\vec u(\vec u^*\cdot J\vec u)+\textbf{A}'\vec u ,\\
&\textbf{A}'=\textrm{diag }(\alpha,\ \beta),\ \alpha=\alpha_2(t)-\alpha_1(t),\ \beta=\alpha_3(t)-\alpha_1(t), 
\end{split}
\]
where $\gamma(x,t)$, $\alpha_i(t)$ are defined by (\ref{E:lax-p}), (\ref{E:lax-q-2}), (\ref{E:lax-q-3}), and (\ref{E:p-a}). By Lemma \ref{L:p10}, we conclude $\gamma(x,t)\equiv 0$. Hence the theorem reduces to showing $\alpha_i$, $i\in\{1,2,3\}$, are certain constants determined by $d_1$ and $d_3$.  This can be seen by  taking the limit of (\ref{E:first-p-j-more}) when $x\to -\infty$ and  $\lambda\to\infty$ since $\alpha_i$ are independent of $x$. Moreover, $m(x,t,\lambda)\to b(x,t)$ as $\lambda\to\infty$ will be shown in Theorem \ref{T:cauchy}. Thus we have
\begin{eqnarray}
&&\left(\begin{array}{ccc} d_1 & 0 & 0 \\ 0 & d_1 & 0 \\ 0 & 0 & d_3 \end{array}\right)\nonumber\\
&=&\lim_{x\to -\infty} b^{-1}\left[\left(\begin{array}{cc} -2\epsilon\vec u^* J\vec u  & 0 \\ 0 & \epsilon\left(J\vec u\vec u^*+\vec u\vec u^* J \right)\end{array}\right)
+\textrm{ diag }(\alpha_1,\alpha_2,\alpha_3)\,\right]b\nonumber\\
&=&
{\begin{cases}
{\left(\begin{array}{ccc} 2\epsilon+\alpha_1 & 0 & 0 \\ 0 & -2\epsilon+\alpha_2 & 0 \\ 0 & 0 & \alpha_3 \end{array}\right),}&\textit{ if $b\in\mathfrak P_1$,}\\
{\left(\begin{array}{ccc} -2\epsilon+\alpha_1 & 0 & 0 \\ 0 & 2\epsilon+\alpha_3 & 0 \\ 0 & 0 & \alpha_2 \end{array}\right),} &\textit{ if $b\in\mathfrak P_2$.} 
\end{cases}}
\label{E:infty-TH}
\end{eqnarray}
  So $\alpha_i$ are constant and $(\alpha,\beta)=(4\epsilon,2\epsilon+d_3-d_1)$ or $(-2\epsilon+d_3-d_1,-4\epsilon)$.
\end{proof}

\begin{remark}\label{R:uniqueness}
Theorem \ref{T:TH-GMV} implies that each second twisted $\frac{U(3)}{U(1)\times U(2)}$-flow, for a fixed $(d_1,d_3)$, gives a GMV solution.  We will prove that different $(d_1,d_3)$'s give different second solutions to the same GMV solution (\ref{E:GMV-modified-TH-GMV}) once $d_3-d_1$'s are equal (Theorem \ref{T:cauchy}). Thus the GMV equation is only part of the constraints in the second twisted $\frac{U(3)}{U(1)\times U(2)}$-flows. However, the twisted $\frac{U(3)}{U(1)\times U(2)}$-hierarchy is still called the (generalized) associated hierarchy for the GMV equation in this paper.
\end{remark}

\begin{remark}\label{R:obstruction}
The obstruction to constructing the GMV equation parametrized by an arbitrary pair $(\alpha,\beta)$ by the second  twisted ${\frac {U(3)}{U(1)\times U(2)}}$-flow  is the commutativity condition  (\ref{E:commute}).
\end{remark}
%%%%%%%%%%%%%%%%%%%%%%%%%%%%%%%%%%%%%%%%%%%%%%%%%%%%%%%%%%%%%%%%%%%%%%%%%%%%%%%%%%%%%%%%%%%%%%%%%%%%%%%%%%%%%%%%%%%%%%%%%%%%%%%%
\section{The direct problem  }\label{S:direct-GMV}
%%%%%%%%%%%%%%%%%%%%%%%%%%%%%%%%%%%%%%%%%%%%%%%%%%%%%%%%%%%%%%%%%%%%%%%%%%%%%%%%%%%%%%%%%%%%%%%%%%%%%%%%%%%%%%%%%%%%%%%%%%%%%%%%%
%%%%%%%%%%%%%%%%%%%%%%%%%%%%%%%%%%%%%%%%%%%%%%%%%%%%%%%%%%%%%%%%%%%%%%%%%%%%%%%%%%%%%%%%%%%%%%%%%%%%%%%%%%%%%%%%%%%%%%%%%%%%%%%%
\subsection{The spectral problem }\label{SS:GMV-spectral}
%%%%%%%%%%%%%%%%%%%%%%%%%%%%%%%%%%%%%%%%%%%%%%%%%%%%%%%%%%%%%%%%%%%%%%%%%%%%%%%%%%%%%%%%%%%%%%%%%%%%%%%%%%%%%%%%%%%%%%%%%%%%%%%%
Let $\sigma_2$, $a$  be defined by (\ref{E:GMV-lax-bab}), 
\begin{equation}\label{E:b}
b=b(x)=\left(
\begin{array}{ccc}
1 & 0 & 0 \\
0 & u & -\bar v \\
0 & v & \bar u
\end{array}\right)\in K_1',
\end{equation}
and $b-1\in\mathbb S$. Consider the spectral problem
\begin{equation}\label{E:GMV-lax-spectral}
\begin{split}
&\partial_ x\Psi=\lambda bab^{-1} \Psi
+\frac \epsilon\lambda\sigma_2(bab^{-1}) \Psi,\\
&\Psi(x,\lambda)e^{-x(\lambda a+\frac \epsilon\lambda\sigma_2(a))}\to 1 \textit{ as $x \to -\infty$.}
\end{split}
\end{equation} 
By introducing the normalizations
\begin{eqnarray}
\Psi(x,\lambda)&=&m(x,\lambda)e^{x(\lambda a+\frac \epsilon\lambda\sigma_2(a))}\label{E:normal}\\
&=&b(x)m'(x,\lambda)e^{x(\lambda a+\frac \epsilon\lambda\sigma_2(a))} \label{E:normal-1},
\end{eqnarray}
 the partial differential equation in (\ref{E:GMV-lax-spectral}) turns into
\begin{eqnarray}
&&\frac{\partial m}{\partial x}=\lambda \left(bab^{-1}m-ma\right)+\frac \epsilon\lambda\left(\sigma_2(bab^{-1})m-m\sigma_2(a)\right),\label{E:lax-normal-bd}\\
&&\frac{\partial m'}{\partial x}=  [\lambda a+\frac \epsilon\lambda\sigma_2(a) ,\,\,m'(x,\lambda) ]+Q(x,\lambda) m'(x,\lambda),\label{E:lax-normal-bd-revised-3}
%&sup_{x\in\mathbb R}\breve{m}(x,\lambda)<\infty, \quad \textit{for admissible $\lambda\in\mathbb C$,}\\
%&\lim_{x\to-\infty}\breve{m}(x,\lambda)= 1,\quad \textit{for admissible $\lambda\in\mathbb C$,}
\end{eqnarray}
with
\begin{eqnarray}
&&Q(x,\lambda)%=Q(b(x),v(x),\lambda)
=\frac \epsilon\lambda\left(b^{-1}\sigma_2(bab^{-1})b-\sigma_2(a)\right)-b^{-1}\frac {\partial b}{\partial x}.\label{E:Q-1}
\end{eqnarray}
%However, we will show that $\breve{m}(x,\lambda)\to 1$ as $|\lambda|\to\infty$ as well. Thus $m(x,\lambda)\to b(x)$ as $|\lambda|\to\infty$. Consequently, $\breve{m}$, instead of $m(x,\lambda)$, is a better  normalized eigenfunction of the orginal eigenfunction $\Psi(x,\lambda)$ for investigation of the inverse scattering problem.

%We are now trying to clarify the definition of "`admissible"' spectral variable $\lambda\in\mathbb C$.
\begin{definition}\label{D:proj}
We define the operator $\mathcal J_\lambda=\mathcal J_{a,\lambda}$ on $gl(n,\mathbb C)$ by
\[\mathcal J_\lambda f=\left[ \lambda a+\frac \epsilon\lambda\sigma_2(a),\,f\right],
\]
and $\pi^\lambda_0$, $\pi^\lambda_\pm$ to be the orthogonal projections of $gl(3,\mathbb C)$ to the $0$--, $\pm$--eigenspaces
    of $Re\, \mathcal J_\lambda=\frac 12(\mathcal J_\lambda+(\mathcal J_\lambda)^*)$. Besides, the characteristic curve of (\ref{E:GMV-lax-bab}) is defined by
\begin{equation}\label{E:character}
\Sigma_a=\left\{\lambda\in\mathbb C|\ \textit{the projections $\pi^\lambda_0$, $\pi^\lambda_\pm$ are not continuous at $\lambda$ }\right\}.
\end{equation}
\end{definition}

A direct direct computation yields the characteristic curve $\Sigma_a$ of (\ref{E:GMV-lax-bab}) is $\mathbb R$. Therefore 
we can follow the argument as that in \cite{MW11} to derive
\begin{theorem}\label{T:existence}
Let  $b(x)\in K_1'$ and $b-1\in\mathbb S$. 
Then there exists a bounded set $Z\subset\mathbb C$, such that $Z\cap (\mathbb C\backslash\mathbb R)$ is discrete in $\mathbb C\backslash\mathbb R$ and for $\forall \lambda\in \mathbb C\backslash(\mathbb R\cup Z)$, there exists uniquely a solution $m(x,\lambda)$ of (\ref{E:lax-normal-bd}) satisfying
\begin{eqnarray}
&&\textit{$m(\cdot,\lambda)$ is bounded for each $\lambda\in \mathbb C\backslash(\mathbb R\cup Z)$},\label{E:lax-normal-bd-revised-1}\\
&&\textit{$m(x,\lambda)\to 1 $ as $x\to -\infty$ for each $\lambda\in \mathbb C\backslash(\mathbb R\cup Z)$,}\label{E:lax-normal-bd-revised-2}\\
&&\textit{$m(x,\cdot)$ is meromorphic in $\mathbb C\backslash\mathbb R$ with poles at $\lambda\in Z$},\label{E:lax-mero}\\
&&\textit{$m(x,\lambda)\to b(x) $ uniformly as $|\lambda|\to\infty$.}\label{E:lax-infty} 
\end{eqnarray}
\end{theorem}

Furthermore we have
\begin{theorem}\label{T:analytic-charact} 
For generic $b(x)$ satisfying the assumption of Theorem \ref{T:existence}, the set $Z$ is a finite set contained in $\mathbb C\backslash\mathbb R$, and $m(x,\lambda)$ has a continuous extension, denoted as $m_\pm(x,\lambda)$, to $\mathbb R$ from $\mathbb C^\pm$. In addition, there exists $V(\lambda)$, $\lambda\in\mathbb R\cup Z$, such that
\begin{eqnarray} 
&&m_+(x,\lambda)=m_-(x,\lambda)e^{x\left(\lambda a+\frac \epsilon{\lambda}\sigma_2(a)\right)}V(\lambda)e^{-x\left(\lambda a+\frac \epsilon{\lambda}\sigma_2(a)\right)},\textit{ $\lambda\in\mathbb R$,} \label{E:jump}\\
&&\textit{$m(x,\lambda)\left(1-\frac{e^{x\left(\lambda a+\frac \epsilon{\lambda}\sigma_2(a)\right)}V(\lambda_0)e^{-x\left(\lambda a+\frac \epsilon{\lambda}\sigma_2(a)\right)}}{\lambda-\lambda_0}\right)$ is regular at $\lambda_0\in Z$,}\label{E:discrete-SD}
\end{eqnarray}and for $\lambda\in\mathbb R$
\begin{equation}\label{E:ana-sd}
\begin{split}
%&|V-1|<<1,\\
&\textit{$\partial_\lambda^\alpha(V-1)$ is $\mathcal O(\lambda^N)$ as $\lambda\to 0$ and $\mathcal O(\lambda^{-N})$ as $|\lambda|\to \infty$} 
\end{split}
\end{equation}
for all positive integer $N$ and nonnegative integer $\alpha$,
\begin{eqnarray}
\textrm{det }V\equiv 1,\label{E:determinant}\\
V(\bar\lambda)^*V(\lambda)^{-1}&=&1,\label{E:self-adjoint-v}\\
\sigma_1(V(-\lambda))V(\lambda)&=&1,\label{E:sigma-1-loop-v}\\
\sigma_2(V(\epsilon/\lambda))V(\lambda)&=&1,\label{E:sigma-2-loop-v}
\end{eqnarray}and for $\lambda\in Z$, 
\begin{eqnarray}
&&V(\lambda)^2=0,\label{E:ds-degenerate}\\
&&V(\lambda)=-V(\bar\lambda)^*,\label{E:ds-selfadj}\\
&&V(\lambda)=-\sigma_1(V(-\lambda)),\label{E:ds-sigma1}\\
&&V(\lambda)=-\frac{\lambda^2}\epsilon\sigma_2(V(\frac\epsilon\lambda)).\label{E:ds-sigma2}
\end{eqnarray}
\end{theorem}
\begin{proof} The generic property for simple pole with residue satisfying (\ref{E:ds-degenerate}) has be shown in \cite{BC84}. The reality conditions (\ref{E:self-adjoint-v})-(\ref{E:sigma-2-loop-v}), (\ref{E:ds-selfadj})-(\ref{E:ds-sigma2}) can be proved by showing
\begin{eqnarray}
m(x,\bar\lambda)^*&=& m(x,\lambda)^{-1},\label{E:self-adjoint}\\
\sigma_1(m(x,-\lambda))&=&m(x,\lambda),\label{E:sigma-1-loop}\\
\sigma_2(m(x,\epsilon/\lambda))&=&m(x,\lambda),\label{E:sigma-2-loop}
\end{eqnarray} and using the properties (\ref{E:jump}), (\ref{E:discrete-SD}), and (\ref{E:ds-degenerate}). Finally, by the same argument as that in the proof of Proposition 2.1 in \cite{BC84}, one can prove  
\begin{equation}\label{E:determinant-psi-m}
\textrm{det }\Psi=\textrm{det }m\equiv 1.
\end{equation}
The statement (\ref{E:determinant}) follows from (\ref{E:determinant-psi-m}) and (\ref{E:jump}).
\end{proof}

\begin{definition}\label{D:scattering-data}
The associated  scattering data of the generic 
potential $b(x)$ is defined by the matrix function $V(\lambda)$, $ \lambda\in\mathbb R\cup Z$, provided $b$ satisfies the assumption of Theorem \ref{T:analytic-charact}. Moreover, $V(\lambda)$ is called a scattering data, if $V(\lambda)$, $ \lambda\in\mathbb R\cup Z$, satisfies (\ref{E:ana-sd})-(\ref{E:ds-sigma2}). 
\end{definition}

%%%%%%%%%%%%%%%%%%%%%%%%%%%%%%%%%%%%%%%%%%%%%%%%%%%%%%%%%%%%%%%%%%%%%%%%%%%%%%%%%%%%%%%%%%%%%%%%%%%%%%%%%%%%%%%%%%%%%%%%%%%%%%%%
\subsection{An extended direct problem  }\label{SS:direct-GMV-ext}
%%%%%%%%%%%%%%%%%%%%%%%%%%%%%%%%%%%%%%%%%%%%%%%%%%%%%%%%%%%%%%%%%%%%%%%%%%%%%%%%%%%%%%%%%%%%%%%%%%%%%%%%%%%%%%%%%%%%%%%%%%%%%%%%%
We need to consider an extended spectral problem of (\ref{E:GMV-lax-spectral}) for solving the inverse problem. The first criteria for an extended spectral problem is preserving the reality conditions with respect to involutions $\sigma_1$, $\sigma_2$ and the self-adjointness. So there could be multiple choices for extended spectral problems. We choose a (splitting type) twisted $\frac{U(4)}{U(2)\times U(2)}$ spectral problem to be our extended system. Since the inverse scattering problem of twisted $\frac{U(4)}{U(2)\times U(2)}$-flows is almost the same as that of the twisted $\frac{O(n,n)}{O(n)\times O(n)}$ which has been tackled in \cite{ABT86}, \cite{MW11}. 

Let $\tilde\sigma_i$, $i=1,\,2 $, be involutions on $U(4)$ defined by
\begin{equation}\label{E:sigma-ext}
\begin{split}
\tilde\sigma_i(x)=\tilde J_{i}x\tilde J_{i}^{-1},& \quad x\in U(4),\\
\tilde J_1=\textrm{diag}(1,1,-1,-1),&\quad \tilde J_2=\textrm{diag}(1,1,-1,1)\\
\end{split}
\end{equation}
and ${ u }(4)=\tilde{\mathcal K}_i \oplus\tilde{\mathcal P}_i$, $i=1,\,2$, 
the Cartan decompositions for $\tilde\sigma_i$. Let $\tilde{\mathcal K}_i$ be the Lie algebras of $\tilde K_i$, i.e.,
\[
\begin{split}
&\tilde K_1=\{\left(\begin{array}{cccc}
a_{11} & a_{12} & 0 & 0 \\
a_{21} & a_{22} & 0 & 0 \\
0 &0 & a_{33} & a_{34}\\
0 & 0 & a_{43} & a_{44}
\end{array}\right)\in U(4) :\left(\begin{array}{cc}
a_{11} & a_{12}  \\
a_{21} & a_{22}\end{array}\right),\,
\left(\begin{array}{cc}
a_{33} & a_{34}\\
a_{43} & a_{44}\end{array}\right)\in U(2)\, \},\\
&\tilde K_2=\{\left(\begin{array}{cccc}
a_{11} & a_{12} & 0 & a_{14} \\
a_{21} & a_{22} & 0 & a_{24} \\
0 &0 & a_{33} & 0\\
a_{41} & a_{42} & 0 & a_{44}
\end{array}\right)\in U(4):\  |a_{33}|=1,\,\left(\begin{array}{ccc}a_{11} & a_{12} & a_{14}\\ 
a_{21} & a_{22} & a_{24}\\
a_{41} & a_{42} & a_{44} \end{array}\right)\in U(3)\,\},\\
&\tilde {\mathcal P}_1=\{\ i\left(\begin{array}{cccc}
0 &0 & u_1 & v_1 \\
0 &0 & u_2 & v_2 \\
u^*_1 & u^*_2 &0 & 0 \\
v^*_1 &v^*_2 & 0 & 0
\end{array}\right)\in u(4) \, \},\quad 
\tilde{\mathcal P}_2=\{\ i\left(\begin{array}{cccc}0 &0& u_1 & 0 \\
0 &0& u_2 & 0 \\ 
u^*_1 &u^*_2 & 0 & u_3 \\ 0 &0& u_3^* & 0 \end{array}\right)\in u(4)\,\}.
\end{split}
\]Let
\[%\label{E:split-K-ext}
\begin{split}
&\tilde S=\{
\left(\begin{array}{cccc}
a_{11} & a_{12} & 0 & 0 \\
a_{21} & a_{22} & 0 & 0 \\
0 &0 & 1 & 0\\
0 & 0 & 0 & 1
\end{array}\right)\in U(4) :\left(\begin{array}{cc}
a_{11} & a_{12}  \\
a_{21} & a_{22}\end{array}\right)\in U(2)\}
\subset\tilde K_1\cap \tilde K_2,\\
&\tilde{\mathcal S}=\{\left(\begin{array}{cccc}
a_{11} & a_{12} & 0 & 0 \\
a_{21} & a_{22} & 0 & 0 \\
0 &0 & 0 & 0\\
0 & 0 & 0 & 0
\end{array}\right)\in u(4) :\left(\begin{array}{cc}
a_{11} & a_{12}  \\
a_{21} & a_{22}\end{array}\right)\in u(2)\,\},
\end{split}
\]
and
\begin{gather}
\tilde K_1=\tilde S\otimes \tilde K_1',\quad \tilde K_1'= 1_{2\times 2} \otimes U(2),\label{E:split-K-1-ext}\\
\tilde{\mathcal K_1}=\tilde{\mathcal S }\oplus \tilde{\mathcal K}_1',\quad \tilde{\mathcal K}_1'= 0_{2\times 2}\oplus u(2).\label{E:split-K-2-ext}\
\end{gather}
%Finally, let
%\begin{equation}\label{E:CSA-ext}
%\tilde{\mathcal A} =\{\ i\left(\begin{array}{cccc}
%0 &0& r & 0 \\
%0 &0& 0 & s \\
%r & 0 & 0 & 0 \\
%0 & s & 0 & 0 
%\end{array}\right)\in u(4) :\ r,\,s \in\mathbb R.\, \}
%\end{equation}
%be a Cartan subalgebra in $su(4)$.

The extended spectral problem of (\ref{E:GMV-lax-spectral}) is
\begin{equation}\label{E:GMV-lax-spectral-ext}
\begin{split}
&\partial_ x\tilde \Psi=\lambda \tilde b\tilde a_1\tilde b^{-1} \tilde \Psi
+\frac \epsilon\lambda\tilde\sigma_2(\tilde b\tilde a_1\tilde b^{-1}) \tilde \Psi,\\
&\tilde\Psi(x,\lambda)e^{-x(\lambda \tilde a_1+\frac \epsilon\lambda\tilde \sigma_2(\tilde a_1))}\to 1 \textit{ as $x \to -\infty$,}
\end{split}
\end{equation}
with
\begin{equation}\label{E:GMV-lax-bab-ext}
\begin{split}
&\tilde a_1=i\left(\begin{array}{cccc}
0 & 0& 1 & 0\\
0 & 0 & 0& 0\\
 1& 0 & 0 & 0\\
 0& 0 & 0 & 0
\end{array}
\right) ,\quad
\tilde b(x,t)=\left(\begin{array}{rrrr}
1 & 0& 0 & 0\\
0 & 1 & 0& 0\\
 0& 0 & u & -\bar v\\
 0& 0 & v & \bar u
\end{array}
\right)\in  \tilde K_1'.
\end{split}
\end{equation}

We note that $\tilde a_1$ is not an oblique direction for solving the twisted $\frac {U(4)}{U(2)\times U(2)}$-spectral problem (cf \cite{ABT86}, \cite{MW11}). However, for the extended spectral problem (\ref{E:GMV-lax-spectral-ext}), (\ref{E:GMV-lax-bab-ext}), we have
\begin{lemma}\label{L:eigenfunction-ext}
The spectral equation (\ref{E:GMV-lax-spectral}) is satisfied by $\Psi(x,\lambda)=(\Psi_{ij})_{1\le i,\,j\le 3}$ if and only if (\ref{E:GMV-lax-spectral-ext}) is satisfied by
\begin{equation}\label{E:tilde-psi}
\tilde\Psi(x,\lambda)=
\left(\begin{array}{cccc}
\Psi_{11} & 0& \Psi_{12} & \Psi_{13}\\
0 & 1 & 0& 0\\
\Psi_{21}& 0 & \Psi_{22} & \Psi_{23}\\
\Psi_{31}& 0 & \Psi_{32} & \Psi_{33}
\end{array}
\right).
\end{equation}Moreover, let $m=(m_{ij})_{1\le i,\,j\le 3}$, $m'=(m')_{1\le i,\,j\le 3}$ be the normalized eigenfunctions defined by (\ref{E:normal}) and (\ref{E:normal-1}) and 
\begin{eqnarray}
\tilde\Psi(x,\lambda)&=&\tilde m(x,\lambda)e^{x(\lambda \tilde a_1+\frac \epsilon\lambda\tilde\sigma_2(\tilde a_1))}\label{E:normal-ext}\\
&=&\tilde b(x)\tilde m'(x,\lambda)e^{x(\lambda \tilde a_1+\frac \epsilon\lambda\tilde \sigma_2(\tilde a_1))} \label{E:normal-1-ext}.
\end{eqnarray}Then 
\begin{equation}\label{E:tilde-m}
\tilde m(x,\lambda)=
\left(\begin{array}{cccc}
m_{11} & 0& m_{12} & m_{13}\\
0 & 1 & 0& 0\\
m_{21}& 0 & m_{22} & m_{23}\\
m_{31}& 0 & m_{32} & m_{33}
\end{array}
\right),\quad \tilde m'(x,\lambda)=
\left(\begin{array}{cccc}
m'_{11} & 0& m'_{12} & m'_{13}\\
0 & 1 & 0& 0\\
m'_{21}& 0 & m'_{22} & m'_{23}\\
m'_{31}& 0 & m'_{32} & m'_{33}
\end{array}
\right).
\end{equation}Finally, for generic $b$,  there exists a finite set $Z\subset \mathbb C\backslash\mathbb R$ and %corresponding (\ref{E:lax-normal-bd-revised-1}), (\ref{E:lax-normal-bd-revised-2}),
\begin{eqnarray} 
&&\tilde m_+(x,\lambda)=\tilde m_-(x,\lambda)e^{x\left(\lambda \tilde a_1+\frac \epsilon{\lambda}\tilde \sigma_2(\tilde a_1)\right)}\tilde V(\lambda)e^{-x\left(\lambda \tilde a_1+\frac \epsilon{\lambda}\tilde \sigma_2(\tilde a_1)\right)},\ \lambda\in \mathbb R,\label{E:jump-ext}\\
&&\tilde m(x,\lambda)\left(1-\frac {e^{x\left(\lambda \tilde a_1+\frac \epsilon{\lambda}\tilde \sigma_2(\tilde a_1)\right)}\tilde V(\lambda_0)e^{-x\left(\lambda \tilde a_1+\frac \epsilon{\lambda}\tilde \sigma_2(\tilde a_1)\right)}}{\lambda-\lambda_0}\right)
\textit{ is regular at $\lambda_0\in Z$.}\label{E:discrete-ext}
\end{eqnarray}
%(\ref{E:bdy-spectral}), (\ref{E:jump-charact}), (\ref{E:bdy-spectral-1-infty}), (\ref{E:bdy-spectral-1-0}), (\ref{E:self-adjoint}) - (\ref{E:sigma-2-loop}),  
with
\begin{equation}\label{E:tilde-v}
\tilde V(\lambda)=
\left(\begin{array}{cccc}
V_{11} & 0& V_{12} & V_{13}\\
0 & 1 & 0& 0\\
V_{21} & 0 & V_{22} & V_{23}\\
V_{31}& 0 & V_{32} & V_{33}
\end{array}
\right)
\end{equation} and for $\lambda\in\mathbb R$,
\begin{equation}\label{E:ana-sd-ext}
\begin{split}
%&|\tilde V-1|<<1,\\
&\textit{$\partial_\lambda^\alpha(\tilde V-1)$ is $\mathcal O(\lambda^N)$ as $\lambda\to 0$ and $\mathcal O(\lambda^{-N})$ as $|\lambda|\to \infty$} 
\end{split}
\end{equation}
\begin{eqnarray}
\textrm{det }\tilde V\equiv 1,\label{E:determinant-ext}\\
\tilde V(\bar\lambda)^*\tilde V(\lambda)^{-1}&=&1,\label{E:self-adjoint-v-ext}\\
\tilde \sigma_1(\tilde V(-\lambda))\tilde V(\lambda)&=&1,\label{E:sigma-1-loop-v-ext}\\
\tilde \sigma_2(\tilde V(\epsilon/\lambda))V(\lambda)&=&1,\label{E:sigma-2-loop-v-ext}
\end{eqnarray}
and for $\lambda\in\mathbb Z$,
\begin{eqnarray}
&&\tilde V(\lambda)^2=0,\label{E:deg-sd-ext}\\
&&\tilde V(\lambda)=-\tilde V(\bar\lambda)^*,\label{E:self-adjoint-sd-ext}\\
&&\tilde V(\lambda)=-\sigma_1(\tilde V(-\lambda)),\label{E:sigma-1-sd-ext}\\
&&\tilde V(\lambda)=-\frac{\lambda^2}\epsilon\sigma_2(\tilde V(\frac\epsilon\lambda)).\label{E:sigma-2-sd-ext}
\end{eqnarray}

%Conversely, if (\ref{E:jump-ext}) - (\ref{E:self-adjoint-sd-ext}) are satisfied, then (\ref{E:tilde-m}) holds.
\end{lemma}
\begin{proof} The statements can be proved by a direct computation. 
\end{proof}
\begin{definition}\label{D:scattering-data-ext}
The associated extended scattering data of $b$ is defined by the matrix function $\tilde V(\lambda)$, $ \lambda\in\mathbb R\cup Z$, provided $b$ satisfies the assumption of Theorem \ref{T:analytic-charact}. Moreover,  $\tilde V(\lambda)$ is called an extended scattering data, if $\tilde V(\lambda)$, $ \lambda\in\mathbb R\cup Z$, satisfies (\ref{E:tilde-v})-(\ref{E:sigma-2-sd-ext}).
\end{definition}

\begin{remark}\label{R:remarks}
If $b(x)\in K_1'$, $b(x)-\left(\begin{array}{ccc}1 & 0 & 0\\ 0 & 0 & -1 \\ 0 & 1 & 0\end{array}\right)\in\mathbb S$, then the spectral problem needed to be considered is
\begin{equation}\label{E:GMV-lax-spectral-case2}
\begin{split}
&\partial_ x\Psi=\lambda bab^{-1} \Psi
+\frac \epsilon\lambda\sigma_2(bab^{-1}) \Psi,\\
&\Psi(x,\lambda)e^{-x(\lambda a+\frac \epsilon\lambda\sigma_2(a))}\to \left(\begin{array}{ccc}1 & 0 & 0\\ 0 & 0 & -1 \\ 0 & 1 & 0\end{array}\right)
 \textit{ as $x \to -\infty$.}
%&b(x)\to  \textit{ as $|x|\to\infty$,}
\end{split}
\end{equation}
It is more convenient to use a change of variables to turn (\ref{E:GMV-lax-spectral-case2}) into
\begin{equation}\label{E:GMV-lax-spectral-change}
\begin{split}
&\partial_ x\Psi=\lambda b'a'{b'}^{-1} \Psi
+\frac \epsilon\lambda\sigma_2(b'a'{b'}^{-1}) \Psi,\\
&\Psi(x,\lambda)e^{-x(\lambda a+\frac \epsilon\lambda\sigma_2(a))}\to 1 \textit{ as $x \to -\infty$,}
\end{split}
\end{equation}
with
\begin{eqnarray}
&&a'=i\left(\begin{array}{ccc}
0 & 0 & 1\\
0 & 0 & 0\\
1 & 0 & 0
\end{array}
\right) \in \mathcal P_1,\quad
b'(x,t)=\left(\begin{array}{ccc}
1 & 0 & 0\\
 0 & \bar v & u\\
 0 & -\bar u & v
\end{array}
\right)\in  K_1'.\label{E:GMV-lax-bab'}
\end{eqnarray}
By analogy, one can derive the existence theorem of the eigenfunction $\Psi(x,\lambda)$, extract continuous and discrete scattering data, and solve the associated extended direct problem
\begin{equation}\label{E:GMV-lax-spectral-ext-varied}
\begin{split}
&\partial_ x\tilde \Psi=\lambda {\tilde {b'}}\tilde a_2{\tilde {b'}}^{-1} \tilde \Psi
+\frac \epsilon\lambda\tilde\sigma_2({\tilde {b'}}\tilde a_2{\tilde {b'}}^{-1}) \tilde \Psi,\\
&\tilde\Psi(x,\lambda)e^{-x(\lambda \tilde a+\frac \epsilon\lambda\tilde \sigma_2(\tilde a))}\to 1 \textit{ as $x \to -\infty$,}
\end{split}
\end{equation}
with
\begin{equation}\label{E:GMV-lax-bab-ext-varied}
\begin{split}
&\tilde a_2=i\left(\begin{array}{cccc}
0 & 0& 0 & 0\\
0 & 0 & 0& 1\\
 0& 0 & 0 & 0\\
 0& 1 & 0 & 0
\end{array}
\right) ,\quad
{\tilde {b'}}(x,t)=\left(\begin{array}{rrrr}
1 & 0& 0 & 0\\
0 & 1 & 0& 0\\
 0& 0 & \bar v & u\\
 0& 0 & -\bar u & v
\end{array}
\right)\in  \tilde K_1',
\end{split}
\end{equation} 
and define the extended scattering data $
\tilde V(\lambda)$, $\lambda\in\mathbb R\cup Z$.

These two different boundary conditions, (\ref{E:GMV-lax-spectral}) and (\ref{E:GMV-lax-spectral-case2}), are the only cases which can be tackled by our approach. Since the associated spectral operators are perturbation of diagonalizable operators. That is, $a$ and $\sigma(a)$ ($a'$ and $\sigma(a')$ respectively) can be simultaneously diagonalized.
\end{remark}

%%%%%%%%%%%%%%%%%%%%%%%%%%%%%%%%%%%%%%%%%%%%%%%%%%%%%%%%%%%%%%%%%%%%%%%%%%%%%%%%%%%%%%%%%%%%%%%%%%%%%%%%%%%%%%%%%%%%%%%%%%%%%%%%
\section{The inverse problem  }\label{S:inverse-GMV}
%%%%%%%%%%%%%%%%%%%%%%%%%%%%%%%%%%%%%%%%%%%%%%%%%%%%%%%%%%%%%%%%%%%%%%%%%%%%%%%%%%%%%%%%%%%%%%%%%%%%%%%%%%%%%%%%%%%%%%%%%%%%%%%%%
In this section, the normalized eigenfunction ${\tilde m}'(x,\lambda)$ will be constructed from the scattering data by solving a Riemann-Hilbert problem.  To find the gauge $\tilde b(x)$ to reconstruct  $\tilde m(x,\lambda)$, one needs to understand the symmetries between coefficients of $(\partial_x{\tilde m}')({\tilde m}')^{-1}$ which is equivalent to solving an over-determined differential systems. 
Inspired by the result of \cite{ABT86}, \cite{MW11}, we reconstruct the gauge via solving an exterior  differential system  
derived from one-dimensional systems associated with Cartan subalgebras with higher ranks (cf. Definition 3.1 in \cite{MW11}). This is the motivation for us to study the extended twisted $\frac{U(4)}{U(2)\times U(2)}$-spectral problem in $\S$ \ref{SS:direct-GMV-ext}. 

The major differences between loop algebra structures associated with twisted $\frac{O(2,2)}{O(2)\times O(2)}$- and with twisted $\frac{U(4)}{U(2)\times U(2)}$-hierarchies are the symmetric and antisymmetric properties of ${\mathcal P_0}$ and $\tilde {\mathcal P}_1$. However, the proof of the inverse problem in Section 6 of \cite{MW11} mainly involves with the involution properties of $ \sigma_i$, the commutativity property $[ a_i, a_j]=0$ and the self-adjointness of $ K_0$, and  has nothing to do with the symmetric property of ${\mathcal P_0}$. As a result, the inverse scattering problem of twisted $\frac {U(4)}{U(2)\times U(2)}$-hierarchy can be solved by the same argument. We will  state the results,  leave analogous details to \cite{ABT86}, \cite{MW11}, and only give the proof for projecting the extended inverse results to that of a twisted $\frac{U(3)}{U(1)\times U(2)}$-spectral problem in this section.

Write $\tilde a=\left(\begin{array}{cc} 0 & D\\ -D^* &0\end{array}\right)$, and $D=\textsl{diag }(w_1,w_2)$. Define
\begin{gather}
\vec x=(x_1,x_2)=x(w_1,w_2),\label{E:vecx}\\
 X=x_1\tilde a_1 + x_2\tilde a_2,\label{E:vecX}\\
\tilde a_1=
i\left(\begin{array}{cccc}
0 & 0 & 1 & 0\\
0 & 0 & 0 & 0\\
1 & 0  & 0 & 0\\
0 & 0 & 0 & 0
\end{array}
\right),\quad
\tilde a_2=
i\left(\begin{array}{cccc}
0 & 0 & 0 & 0\\
0 & 0 & 0 & 1\\
0 & 0 & 0 & 0\\
0 & 1 & 0 & 0
\end{array}
\right).\label{E:a12}
%,\\M(\vec x,\lambda)=\tilde m'(x,a,\lambda)=\tilde m'(x,\lambda).\label{E:M}
\end{gather}
\begin{theorem}\label{T:inverse-m'}
Let  $\tilde V(\lambda)$, $\lambda\in\mathbb R$ satisfy
the analytical constraints (\ref{E:ana-sd-ext}) and the algebraic constraints (\ref{E:determinant-ext})-(\ref{E:sigma-2-sd-ext}). Then there exists uniquely $M(\vec x,\lambda)$ such that
\begin{eqnarray}
&&M_+(\vec x,\lambda)=M_-(\vec x,\lambda)e^{\lambda X+\frac \epsilon{\lambda}\tilde\sigma_2(X)}\tilde V(\lambda)e^{-\lambda X-\frac \epsilon{\lambda}\tilde\sigma_2(X)},\quad\lambda\in\mathbb R,\label{E:RH-m'}\\
&&M(\vec x,\lambda)\left(1-\frac {e^{\lambda X+\frac \epsilon{\lambda}\tilde\sigma_2(X)}\tilde V(\lambda_0)e^{-(\lambda X+\frac \epsilon{\lambda}\tilde\sigma_2(X))}}{\lambda-\lambda_0}\right)
\textit{is regular at $\lambda_0\in Z$,}\label{E:RH-sd}\\
&&\textit{$M(\vec x,\lambda)$ is holomorphic for $\lambda\in\mathbb C\backslash (\mathbb R\cup Z)$, $M(\vec x,\lambda)\to 1$ as $|\lambda|\to\infty$},\label{E:RH-bdry-m'}
\end{eqnarray}
and $M(\vec x,\lambda)$ satisfies the analytical and algebraic conditions
\begin{eqnarray}
&&\textrm{det }M\equiv 1,\label{E:determinant-m-inv}\\
&&M(\vec x,\bar\lambda)^*=M(\vec x,\lambda)^{-1},\label{E:self-adjoint-inv}\\
&&\tilde \sigma_1(M(\vec x,-\lambda))=M(\vec x,\lambda),\quad
 \tilde \sigma_2(M^{-1}(\vec x,0)M(\vec x,\frac \epsilon\lambda))=M(\vec x,\lambda)\label{E:sigma-12-inv}\\
 &&\textit{$x^k\partial^{k'}_x\lambda^{k'}(M(\vec x,\lambda)-1)\in L^2(\mathbb R)$ for $\forall k,\,k'$ and tends to $0$ uniformly }\label{E:RH-reg-m'}\\
&&\textit{as $x\to -\infty$}; \textit{ $\exists \delta(\lambda)$ diagonal, s.t. $x^k\partial^{k'}_x\lambda^{k'}(M(\vec x,\lambda)-\delta(\lambda))\in L^2(\mathbb R)$}\nonumber\\
&&\textit{for $\forall k,\,k'$ and tends to $0$ uniformly as $x\to \infty$.}\nonumber
\end{eqnarray}
Moreover, if $\tilde V(\lambda)$ is an extended scattering data, i.e. is of the form (\ref{E:tilde-v}), then
\begin{equation}\label{E:M-extend}
M((x,0),\lambda)=
\left(\begin{array}{cccc}
M_{11} & 0& M_{12} &M_{13}\\
0 & 1 & 0& 0\\
M_{21}& 0 & M_{22} & M_{23}\\
M_{31}& 0 & M_{32} & M_{33}
\end{array}
\right).
\end{equation}
\end{theorem}
\begin{proof} The existence of $M(\vec x,\lambda)$ satisfying (\ref{E:RH-m'})-(\ref{E:RH-bdry-m'}), and (\ref{E:RH-reg-m'}) can be proved by the same argument as that in the proof of Theorem 5.1 in \cite{MW11}. Properties (\ref{E:determinant-ext}), (\ref{E:deg-sd-ext}) imply $\det{M}$ is continuous for $\lambda\in\mathbb C$.  Thus Condition (\ref{E:determinant-m-inv}) is shown by noting $\partial_{\bar\lambda} \det {M}=0$ for $\lambda\in\mathbb C^\pm$  and applying Liouville's theorem. The statements (\ref{E:self-adjoint-inv}), (\ref{E:sigma-12-inv}), and (\ref{E:M-extend}) can be proved by the uniqueness property of $M(\vec x,\lambda)$. 
\end{proof}

Defining the  asymptotic expansions
\begin{eqnarray}
M( \vec x,\lambda)\to & 1+\sum_{k=1}^\infty M_k^{\sharp}(\vec x)\lambda^{-k} &\textit{ as $|\lambda|\to \infty$,}\label{E:m'infty}\\
M(\vec x,\lambda)\to &\sum_{k=0}^\infty M_k^{\flat}(\vec x)\lambda^{k} &\textit{ as $|\lambda|\to 0$.}\label{E:m'0}
\end{eqnarray} and applying the same argument as that in the proof of Lemma 6.1 - 6.3, and Theorem 6.1 in \cite{MW11}, one can derive the following four lemmas. 
\begin{lemma}\label{L:pdem'}
Suppose $M(\vec x,\lambda)$ is derived by Theorem \ref{T:inverse-m'}. Then
\begin{equation}
\frac {\partial M}{\partial x_j}=[\lambda \tilde a_j+\frac \epsilon\lambda\tilde \sigma_2(\tilde a_j), M ]+\frac \epsilon\lambda\left(B_j(\vec x)-\tilde \sigma_2(\tilde a_j)\right)M-C_j(\vec x)M,\label{E:M'}
\end{equation}
with
\begin{gather}
    B_j(\vec x)\in \tilde{\mathcal P}_1\cap C^\infty,\quad C_j(\vec x)\in \tilde{\mathcal K}_1\cap C^\infty.\label{E:BC}
\end{gather}
\end{lemma}

\begin{lemma}\label{L:CC}
The compatibility conditions of (\ref{E:M'}) are
\begin{equation}
\partial_{x_j}C_i-\partial_{x_i}C_j-\left[C_i,C_j\right]=
\epsilon\left[\tilde a_i, B_j\right]-\epsilon\left[\tilde a_j, B_i\right].\label{E:cc}
\end{equation}
\end{lemma}

\begin{lemma}\label{L:gauge}
For any constant $\mu\in\mathbb R$, there exists uniquely
\begin{equation}
\tilde b(\vec x) \in \tilde K_1'\cap C^\infty, \label{E:gauge}
\end{equation}
such that 
\begin{eqnarray}
\tilde b(x(w_1,w_2))\to \left(\begin{array}{cc} 1_{2\times 2} & 0 \\ 0 & e^{-i\mu/2} 1_{2\times 2}\end{array}\right) \in \tilde K_1'&&\textit{as $x\to -\infty$,} \label{E:gauge-infty}\\
-\tilde bC_j{\tilde b}^{-1}+(\partial_j\tilde b){\tilde b}^{-1}\in\tilde{\mathcal S} &&\textit{for $\forall j$.}\label{E:gauge-1}
\end{eqnarray}
\end{lemma}
\begin{proof} We remark the boundary condition (\ref{E:gauge-infty}) can be chosen for arbitrary element in $\tilde K_1'$.
\end{proof}

\begin{lemma}\label{L:inv}
Suppose the assumption of Theorem \ref{T:inverse-m'} holds. For any constant $\mu\in\mathbb R$, let
\begin{equation}
\tilde\Psi(x,\lambda)=\tilde b(\vec x)M(\vec x,\lambda)e^{\lambda X+\frac \epsilon\lambda\sigma_2(X)} %=b^{-1}(x)M(x,\lambda)e^{-x(\lambda a+\frac 1\lambda\sigma_1(a))}.
\label{E:definepsi}%\\
%&=&b(x)\breve{m}(x,\lambda)e^{-x(\lambda a+\frac 1\lambda\sigma_1(a))}\nonumber
\end{equation}
Here $x$, $\vec x$, $X$, $M$ satisfy (\ref{E:vecx})-(\ref{E:a12}). Then
\begin{gather}
\frac{\partial \tilde\Psi}{\partial x}=\lambda \tilde b\tilde a{\tilde b}^{-1}\tilde\Psi+\frac \epsilon\lambda\tilde\sigma_2(\tilde b\tilde a{\tilde b}^{-1})\tilde\Psi+\tilde v\tilde\Psi,\label{E:inv-eigen}
\end{gather}
with
\begin{eqnarray}
% v(\vec x)&=&-\frac{\partial s}{\partial x}{s}^{-1},\quad s\in S\cap C^\infty,\label{E:bB}\\
\tilde v(\vec x)&=&\sum_{j=1}^2w_j(-\tilde bC_j{\tilde b}^{-1}+(\partial_{x_j}\tilde b){\tilde b}^{-1})\in\tilde{\mathcal S}\cap \mathbb S,\label{E:C}
\end{eqnarray}
where $C_j$, $\tilde b(\vec x)$ are defined by Lemma \ref{L:pdem'}, \ref{L:gauge}, respectively. Moreover, 
\begin{gather}
\tilde b(\vec x)-\left(\begin{array}{cc} 1_{2\times 2} & 0 \\ 0 & e^{-i\mu/2} 1_{2\times 2}\end{array}\right)\in \tilde K_1'\cap\mathbb S,\label{E:b-schwartz}\\
\textit{$\tilde v$ is independent of $\mu$ defined by (\ref{E:gauge-infty}).}\label{E:invariant}
\end{gather}
\end{lemma}
\begin{proof}
Once  the boundary condition (\ref{E:gauge-infty}) is a diagonal element in $\tilde K_1'$, the argument in proving Theorem 6.1 in \cite{MW11} works well in proving all statements  in Lemma \ref{L:inv} except (\ref{E:invariant}). The property (\ref{E:invariant}) follows from the fact that changing $\mu$ could only alter the $\tilde K_1'$ part of the right hand side of (\ref{E:C}) and $\tilde v\in\tilde {\mathcal S}$.

Note (\ref{E:inv-eigen}) and (\ref{E:definepsi}) imply
\begin{equation}\label{E:M-inv-sch}
\frac{\partial M}{\partial x}=  [\lambda \tilde a+\frac \epsilon\lambda\tilde \sigma_2(\tilde a) ,\,\,M(x,\lambda) ]+Q(x,\lambda) M(x,\lambda)
\end{equation}
\begin{eqnarray}
&&Q(x,\lambda)%=Q(b(x),v(x),\lambda)
=\frac \epsilon\lambda\left(\tilde b^{-1}\tilde \sigma_2(\tilde b\tilde a\tilde b^{-1})\tilde b-\tilde \sigma_2(\tilde a)\right)-\tilde b^{-1}\frac {\partial \tilde b}{\partial x}+\tilde b^{-1}\tilde v\tilde b.\label{E:Q-1-sch}
\end{eqnarray}
Thus the Schwartz properties (\ref{E:C}) and (\ref{E:b-schwartz}) follow from (\ref{E:RH-reg-m'}), (\ref{E:M-inv-sch}), and (\ref{E:Q-1-sch}).
\end{proof}

Lemma \ref{L:inv} solves the inverse problem of a general twisted $\frac {U(4)}{U(2)\times U(2)}$-spectral problem for scattering data $\tilde V(\lambda)$ satisfying (\ref{E:ana-sd-ext})-(\ref{E:sigma-2-sd-ext}). 
The following theorem says that  when $\tilde V(\lambda)$ is an extended scattering data, the above result can be projected to be a solvability of the inverse problem for a twisted $\frac {U(3)}{U(1)\times U(2)}$-spectral problem.
\begin{theorem}\label{T:inv}
Suppose the assumption of Theorem \ref{T:inverse-m'} holds for either
\begin{equation}\label{E:tilde-v-reduced}
\tilde V(\lambda)=
\left(\begin{array}{cccc}
V_{11} & 0& V_{12} & V_{13}\\
0 & 1 & 0& 0\\
V_{21} & 0 & V_{22} & V_{23}\\
V_{31}& 0 & V_{32} & V_{33}
\end{array}
\right),
\end{equation}or
\begin{equation}\label{E:tilde-v-reduced-varied}
\tilde V(\lambda)=
\left(\begin{array}{cccc}
1 & 0& 0 & 0\\
0 & V_{11} & V_{12}& V_{13}\\
0 & V_{21} & V_{22} & V_{23}\\
0 & V_{31} & V_{32} & V_{33}
\end{array}
\right).
\end{equation}Then there exist a unique $\Psi(x,\lambda)\in L_+^\epsilon$ and a unique $b(x)\in K'_1$ satisfying $b(x)- 1\in\mathbb S$ (or $b(x)-\left(\begin{array}{ccc} 1 & 0 & 0 \\ 0 & 0 & -1 \\ 0 &1 & 0\end{array}\right)\in\mathbb S$) such that
\begin{equation}\label{E:inv-eigen-reduced}
\frac{\partial \Psi}{\partial x}=\lambda b a{ b}^{-1}\Psi+\frac \epsilon\lambda\sigma_2( b a{ b}^{-1})\Psi 
\end{equation}
with $a$ defined by (\ref{E:GMV-lax-bab}), and the associated extended scattering data of $b(x)$ is $\tilde V(\lambda)$.
\end{theorem}
\begin{proof} We only prove the case (\ref{E:tilde-v-reduced}). Case (\ref{E:tilde-v-reduced-varied}) can be argued by analogy. 
Let $M(\vec x,\lambda)$, $\tilde b(\vec x)$, $\tilde \Psi(x,\lambda)$ be derived from  Theorem \ref{T:inverse-m'}, Lemma \ref{L:gauge} and \ref{L:inv} by specially choosing $(w_1,w_2)=(1,0)$ in (\ref{E:vecx}), i.e. $\vec x=(x,0)$, $ X=x\tilde a_1 $, $\tilde a=\tilde a_1 \in \tilde{\mathcal A}$. 
Define
\begin{equation}\label{E:M}
M(\vec x,\lambda)=\tilde m'(x,\lambda)
\end{equation}
%Then conditions (\ref{E:RH-m'}), (\ref{E:RH-bdry-m'}), (\ref{E:self-adjoint-inv}), $\tilde a=\tilde a_1 $, 
Applying Theorem \ref{T:inverse-m'} and (\ref{E:tilde-v-reduced}), we have 
\[
\tilde m'(x,\lambda)=
\left(\begin{array}{cccc}
m'_{11} & 0& m'_{12} & m'_{13}\\
0 & 1 & 0& 0\\
m'_{21}& 0 & m'_{22} & m'_{23}\\
m'_{31}& 0 & m'_{32} & m'_{33}
\end{array}
\right).
\]
Together with $\tilde b\in\tilde K_1'$, we find $\tilde v$ in(\ref{E:inv-eigen}) is of the form $\textrm{diag }(i\nu,0,0,0)$. Hence (\ref{E:inv-eigen}) can be gauged to
\begin{gather}
\frac{\partial \tilde\Psi_1}{\partial x}=\lambda \tilde b_1\tilde a{\tilde b_1}^{-1}\tilde\Psi_1+\frac \epsilon\lambda\tilde\sigma_2(\tilde b_1\tilde a{\tilde b}^{-1}_1)\tilde\Psi_1,\label{E:eigenfunction-delete-v}\\
 \tilde b_1=\tilde b\cdot \textrm{diag }(1,1,e^{i\nu},1)\in \tilde K_1'=1\otimes U(2) ,
\label{E:eigenfunction-delete-v-2}\\
\tilde\psi_1=\textrm{diag }(e^{i\nu},1,1,1)\tilde \Psi.\label{E:eigenfunction-delete-v-3}
\end{gather}
Applying the same argument as that in Proposition 2.1 in \cite{MW11} to (\ref{E:eigenfunction-delete-v}), we obtain that $\det(\tilde \Psi_1)$ is constant. Together with 
(\ref{E:determinant-m-inv}), (\ref{E:definepsi}), and (\ref{E:eigenfunction-delete-v-3}), we conclude $e^{i\nu}\det \tilde b$ is constant which equals to $e^{i(\nu(x=-\infty)-\mu)}$ by (\ref{E:gauge-infty}) and (\ref{E:eigenfunction-delete-v-3}). Equation (\ref{E:eigenfunction-delete-v-2}) then yields
\[
\textit{$\det\tilde b_1=e^{i\nu}\det \tilde b\equiv e^{i(\nu(x=-\infty)-\mu)}$}.
\]
Besides,  noting $\nu(x=-\infty)$ exists (by (\ref{E:C})) and is determined by $\tilde v$ (independent of $\mu$). Consequently 
\begin{equation}\label{E:tildeb1}
\tilde b_1=\left(\begin{array}{cc}
1_{2\times 2} & 0 \\
0 & \omega
\end{array}
\right),\quad\omega\in SU(2),\quad \tilde b_1-1\in\mathbb S
\end{equation}  by choosing $\mu=\nu(x=-\infty)$ and using  (\ref{E:invariant}), (\ref{E:eigenfunction-delete-v-2}). 

However, by solving the direct problem of (\ref{E:eigenfunction-delete-v}) with $\tilde b_1$ satisfying (\ref{E:tildeb1}) and applying (\ref{E:lax-infty}), (\ref{E:RH-bdry-m'}), (\ref{E:eigenfunction-delete-v-2}), and (\ref{E:eigenfunction-delete-v-3}), as matter of fact,  $\nu=0$.

Therefore the theorem is proved by defining  
\begin{gather}\label{E:renormalized-reduced}
\Psi(x,\lambda)= b(x) m'(x,\lambda)e^{x(\lambda  a+\frac \epsilon\lambda \sigma_2(  a))}
\end{gather}
with  $a$  defined by (\ref{E:GMV-lax-bab}), $m'(x,\lambda)=(m'_{ij})$, and $ b=\left(\begin{array}{cc}
1_{1\times 1} & 0 \\
0 & \omega
\end{array}
\right)\in K_1'$.

\end{proof}

%%%%%%%%%%%%%%%%%%%%%%%%%%%%%%%%%%%%%%%%%%%%%%%%%%%%%%%%%%%%%%%%%%%%%%%%%%%%%%%%%%%%%%%%%%%%%%%%%%%%%%%%%%%%%%%%%%%%%%%%%%%%%%%%
\section{The Cauchy problem  }\label{S:cauchy-GMV}
%%%%%%%%%%%%%%%%%%%%%%%%%%%%%%%%%%%%%%%%%%%%%%%%%%%%%%%%%%%%%%%%%%%%%%%%%%%%%%%%%%%%%%%%%%%%%%%%%%%%%%%%%%%%%%%%%%%%%%%%%%%%%%%%%
We first apply the inverse scattering theory established in Section \ref{S:direct-GMV} and \ref{S:inverse-GMV} to solve the initial value problem of the $k$-th twisted $\frac {U(3)}{U(1)\times U(2)}$-flow.
\begin{theorem}\label{T:cauchy}
Given $d_1$, $d_3\in\mathbb R$, and $b_0(x)\in K_1'$ such that either $b_0-1\in\mathbb S$ or $b_0-\left(\begin{array}{ccc}1 & 0 & 0\\ 0 & 0 &-1\\ 0 & 1 & 0\end{array}\right)\in\mathbb S$. If the scattering data for $b_0(x)$ is generic, then  the initial value problem of the $k$-th twisted $\frac {U(3)}{U(1)\times U(2)}$-flow admits a unique solution $b(x,t)\in \mathfrak P_1$ or $ \mathfrak P_2$. 
More precisely, there uniquely exist $m(x,t,\lambda)\in L_-^\epsilon$ and $\Psi(x,t,\lambda)=m(x,t,\lambda)e^{x\hat J_{1,0}+t\hat J_k}\in L_+^\epsilon$ such that 
\begin{equation}\label{E:lax-unique-twisted}
\left[{\bf L}, {\bf M}\right] =0
\end{equation}
with
\begin{equation}\label{E:cong-twisted}
\begin{split}
&{\bf L}=\partial_x-\frac{\partial\Psi}{\partial x}\Psi^{-1}=\partial_x-(\lambda bab^{-1}+\frac \epsilon\lambda\sigma_2(bab^{-1})) ,\\
&{\bf M}=\partial_t-\frac{\partial\Psi}{\partial t}\Psi^{-1},\\
&b(x,0)=b_0(x),\quad\textit{ $b(x,t)=m(x,t,\infty)\in \mathfrak P_1$ $(\textit{ or }b(x,t)\in \mathfrak P_2\ )$}.
\end{split}
\end{equation}
\end{theorem}
\begin{proof}
We first apply Theorem \ref{T:existence} to solve the eigenfunction of (\ref{E:GMV-lax-spectral}) for $a$, $b(x,0)=b_0(x)$ defined by (\ref{E:GMV-lax-bab}) and (\ref{E:b}). Applying  Definition \ref{D:scattering-data}  and
Theorem \ref{T:analytic-charact}, we obtain the scattering data $ V(\lambda,0)$, $\lambda\in\mathbb R\cup Z$ for the potential $b(x,0)$. Define
\begin{equation}\label{E:cauchy2}
\begin{split}
&V(\lambda,t)=e^{t\hat J_k}V(\lambda,0) e^{-t\hat J_k},\textit{ for $\lambda\in\mathbb R\cup Z$}
\end{split}
\end{equation}
So $V(\lambda, t)$ satisfies the assumption of Theorem \ref{T:inv} and there exist uniquely smooth $m'(x,t,\lambda)\in L_-^\epsilon$, $b(x,t)\in \mathfrak P_1$ (or $b(x,t)\in \mathfrak P_2$) satisfying
\begin{gather}
m(x,t,\lambda)=b(x,t)m'(x,t,\lambda),\label{E:cauchy-1}\\
\Psi(x,t,\lambda)=b(x,t)m'(x,t,\lambda)e^{x\hat J_{1,0}+t\hat J_k},\label{E:cauchy-2}
\end{gather}
and
\begin{gather}\label{E:spectral-x}
\frac{\partial \Psi}{\partial x}(x,t,\lambda)=\lambda bab^{-1}\Psi(x,t,\lambda)+\frac \epsilon\lambda\sigma_2(bab^{-1})\Psi(x,t,\lambda).
\end{gather}
So   ${\bf L}=\frac{\partial \Psi}{\partial x}\Psi^{-1}\in\mathcal L_+^\epsilon$. Therefore,
${\bf M}=\frac{\partial \Psi}{\partial t}\Psi^{-1}\in\mathcal L_+^\epsilon$ and $
\left[{\bf L}, {\bf M}\right] =0$. 
\end{proof}
Consequently, we can solve the initial value problem for the GMV equation for $(\alpha,\beta)=(\alpha,-4\epsilon)$ or $(\alpha,\beta)=(4\epsilon, \beta)$.
\begin{corollary}\label{C:cauchy}
Given  $\epsilon>0$, $\beta\in\mathbb R$, and a (generic) function $\vec u_0(x)-\left(\begin{array}{c} 1\\ 0\end{array}\right)\in \mathbb S$,  
the initial value problem of the GMV equation 
\begin{equation}\label{E:GMV-1-2}
\begin{split}
&i\vec{u}_t=(\vec u_x-\vec u(\vec u^*\cdot \vec u_x))_x+4\epsilon\vec u(\vec u^*\cdot J\vec u)+\textbf{A}\vec u,\\
&\vec u^* \vec u=1,\quad \vec u\in\mathbb C^2,\quad \vec u(x,0)=\vec u_0(x),\\
&J=\textrm{diag}\,(
-1,\,1),\quad \textbf{A}=\textrm{diag}\,(
4\epsilon ,\, \beta ),
\end{split}
\end{equation}
admits one family of global solutions. 
\end{corollary}
\begin{proof}
The solvability follows from setting $\hat J_k$ to be
\begin{gather*}
\hat J_2=i(a^2\lambda^2
-\left(\begin{array}{ccc}2\epsilon+\alpha_1 & 0 & 0 \\ 0 & 2\epsilon+\alpha_1& 0 \\ 0 & 0 & \beta+\alpha_1 \end{array}\right)
+\sigma_2(a^2)(\frac\epsilon\lambda)^{2}),\quad
\alpha_1\in\mathbb R,
\end{gather*}
and applying Theorem \ref{T:TH-GMV} and \ref{T:cauchy}. Different $\alpha_1$'s correspond to different $b(x,t)$'s since the scattering data differ when $t>0$ by (\ref{E:cauchy2}).
\end{proof}

\begin{corollary}\label{C:cauchy-1}
Given $\epsilon>0$, $\alpha\in\mathbb R$, and a (generic) function $\vec u_0(x)-\left(\begin{array}{c} 0\\ 1\end{array}\right)\in \mathbb S$, 
the initial value problem of the GMV equation 
\begin{equation}\label{E:GMV-1-3}
\begin{split}
&i\vec{u}_t=(\vec u_x-\vec u(\vec u^*\cdot \vec u_x))_x+4\epsilon\vec u(\vec u^*\cdot J\vec u)+\textbf{A}\vec u,\\
&\vec u^* \vec u=1,\quad \vec u\in\mathbb C^2,\quad \vec u(x,0)=\vec u_0(x),\\
&J=\textrm{diag}\,(
-1,\,1),\quad \textbf{A}=\textrm{diag}\,(
\alpha ,-4\epsilon ),
\end{split}
\end{equation}
admits  one family of global solutions. 
\end{corollary}
\begin{proof}
Set $\hat J_k$ to be
\begin{gather*}
\hat J_2=i(a^2\lambda^2
-\left(\begin{array}{ccc}-2\epsilon+\alpha_1 & 0 & 0 \\ 0 & -2\epsilon+\alpha_1& 0 \\ 0 & 0 & \alpha+\alpha_1 \end{array}\right)
+\sigma_2(a^2)(\frac\epsilon\lambda)^{2}),\quad
\alpha_1\in\mathbb R,
\end{gather*}
and apply Theorem \ref{T:TH-GMV} and \ref{T:cauchy}. Different $\alpha_1$'s correspond to different $b(x,t)$'s since the scattering data differ when $t>0$ by (\ref{E:cauchy2}).
\end{proof}

%%%%%%%%%%%%%%%%%%%%%%%%%%%%%%%%%%%%%%%%

\end{document}